\documentclass[11pt,draftcls,onecolumn]{IEEEtran}

\usepackage{amssymb}

\usepackage{cite}
\usepackage{graphicx}

\usepackage[cmex10]{amsmath}
\usepackage{array}
\usepackage{comment}

\newtheorem{theorem}{Theorem}
\newtheorem{lemma}{Lemma}
\newtheorem{rem}{Remark}

\begin{document}
%
\title{An Alternative Proof of an Extremal Entropy Inequality}
%
%
%

\author{Sangwoo~Park,~Erchin~Serpedin, and Khalid Qaraqe 
\thanks{Department of Electrical and Computer Engineering, Texas A\&M University, College Station, TX 77843-3128  USA,  e-mail: serpedin@ece.tamu.edu. A small part of this paper was presented at ISIT 2012.}}
\maketitle

\begin{abstract}
This paper first focuses on deriving an alternative approach for proving an extremal entropy inequality (EEI), originally presented in \cite{Extremal:Liu}. The proposed approach does not rely on the channel enhancement technique, and has the advantage that it yields an explicit description of the optimal solution as opposed to the implicit approach of  \cite{Extremal:Liu}.
Compared with the proofs in \cite{Extremal:Liu}, the proposed alternative proof  is also simpler, more direct,  more information-theoretic, and has the additional advantage that it offers a new perspective for establishing novel as well as known challenging results such the capacity of the vector Gaussian broadcast channel, the lower bound of the achievable rate for distributed source coding with a single quadratic distortion constraint, and the secrecy capacity of the Gaussian wire-tap channel. The second part of this paper is devoted to some novel applications of the proposed mathematical results. The proposed mathematical techniques are further exploited to obtain a more simplified proof of the EEI without using the entropy power inequality (EPI),  to build the optimal solution for a special class of broadcasting channels with  private messages and to obtain a mutual information-based performance bound for the mean square-error of a linear Bayesian estimator of  a Gaussian source embedded in an additive noise channel.

\end{abstract}

\begin{IEEEkeywords}
Entropy Power Inequality  (EPI), Extremal Entropy Inequality (EEI), Data Processing Inequality, Channel Enhancement, Broadcast Channel, Wire-Tap Channel, Cramer-Rao bound, Bayesian estimation
\end{IEEEkeywords}

\IEEEpeerreviewmaketitle

\section{Introduction}\label{sec1}

\IEEEPARstart{T}{he} classical entropy power inequality (EPI) was first established by Shannon \cite{Infor:Shannon}. Due to its importance and usefulness,  EPI was proved by several different authors using distinct methods. In \cite{EPI:Stam}, Stam provided the first rigorous proof, and Stam's proof was further simplified by Blachman \cite{EPI:Blachman} and Dembo et al. \cite{EPI:Dembo}, respectively. Verd\'{u} and Guo  proposed a new proof of the EPI based on the I-MMSE concept \cite{EPI:Guo}. Most recently, Rioul proved the EPI  based only on information theoretic quantities \cite{EPI:Rioul}. Before Rioul's proof, most of the reported proofs were based on de Bruijn-type identities and  Fisher information inequality, i.e., the previous proofs were conducted mainly via an  estimation-theoretic approach rather than  an  information-theoretic  approach.

Due to the significance of the EPI, numerous versions of EPIs such as Costa's EPI \cite{EPI:Costa}, the EPI for dependent random variables \cite{EPI:Johnson}, and the extremal entropy inequality (EEI) \cite{Extremal:Liu}  have been proposed. Among the EPIs, the extremal entropy inequality is especially prominent since it can be adapted to several important applications investigated  recently in the wireless communications area. In \cite{Extremal:Liu}, Liu and Viswanath  proposed the extremal entropy inequality, motivated by multi-terminal information theoretic problems such as the vector Gaussian broadcast channel and the distributed source coding with a single quadratic distortion constraint, and suggested several applications for the extremal entropy inequality. The EEI is an entropy power inequality which includes a covariance constraint. Because of the covariance constraint, the EEI could not be proved directly by using the classical EPI. Therefore,   a  powerful technique, referred to as the channel enhancement technique \cite{CapBroad:Shamai}, was adopted in the proofs reported in \cite{Extremal:Liu}.

The proofs proposed in \cite{Extremal:Liu} proceed as follows. First, the extremal entropy inequality is cast as an optimization problem. Using the channel enhancement technique, which relies mainly  on Karush-Kuhn-Tucker (KKT) conditions, an alternative optimization problem, whose maximum value is larger than the maximum value of the original problem, is proposed, and the alternative problem is solved using the EPI. Finally, the proof is completed by showing that the maximum value of the alternative problem is equal to the maximum value of the original problem. Even though Liu and Viswanath proposed two kinds of proofs, a direct proof and a perturbation proof, both proofs are commonly based on the channel enhancement technique, and they are derived in a similar way except de Bruijn's identity is adapted in the perturbation proof.

The main theme of this paper is to develop a novel mathematical framework to prove the extremal entropy inequality without using the channel enhancement technique. Since the channel enhancement technique is adapted to prove not only the extremal entropy inequality but also the capacity of several different kinds of Gaussian channels, e.g., the capacity of the Gaussian broadcast channel and the secrecy capacity of the Gaussian wire-tap channel, by finding an alternative proof for the extremal entropy inequality,  one can also find novel techniques to calculate the capacity of Gaussian broadcast channel, the secrecy capacity of Gaussian wire-tap channel, etc. More important is the fact that the mathematical framework and tools developed in the first part of this paper to achieve an alternative proof of the extremal entropy inequality without using the channel enhancement technique are exploited in the second part of this paper to achieve a second proof of EEI (a more simplified proof of the EEI that does not use neither the EPI  nor the worst additive noise lemma),  to obtain the optimal solution of a special class of  broadcasting problems that assume a private message,   and to characterize the minimum mean-square error (MMSE) performance of linear Bayesian estimators of a Gaussian source in additive noise channels.

The first proof of the EEI,  proposed in the first part of this paper, exploits mainly  four techniques: the data processing inequality, the moment generating function, the worst additive noise lemma, and the classical EPI. By using the data processing inequality, the worst additive noise lemma, and the classical EPI,  an upper bound is calculated. Then, by applying the equality condition of the data processing inequality, we prove that the upper bound can be achieved. The moment generating functions are implemented to prove the achievement of the equality condition in the data processing inequality.  The second proposed proof of the EEI relies partly on the techniques and tools proposed in the first proof of the EEI, and it is further  simplified in the sense that it does not rely neither on the EPI nor on the additive worst noise lemma.

The contributions of our proof can be summarized as follows. In the first part of this paper, a first alternative proof of the EEI is proposed, and it is shown to be simpler and more direct compared with the proofs in \cite{Extremal:Liu}. The proposed proof yields a more information-theoretic approach without using the KKT conditions. The proposed approach relies on the data processing inequality, and the moment generating function helps to circumvent the step of using the KKT conditions. Moreover, by simply analyzing some properties of positive semi-definite matrices, one can bypass the step of proving the existence of the optimal solution which satisfies the KKT conditions, a step which is very complicated to accomplish. In addition, the structure of the covariance matrix of the optimal solution is mentioned in detail by using properties of positive semi-definite matrices. Therefore, the proposed approach yields an explicit description of the optimal solution as opposed to the implicit solutions in \cite{Extremal:Liu}. Furthermore, the proposed proof presents a novel investigation method not only for the extremal inequality but also for  applications such as the capacity of Gaussian broadcast channel, the secrecy capacity of Gaussian wire-tap channel, and so on. In the second part of this paper, the tools and mathematical approach used in the first part of the paper to prove the EEI are further simplified to obtain a second alternative proof of the EEI without using the EPI or the worst additive noise lemma. Two additional applications of the proposed results in finding the optimal signaling scheme for a broadcasting problem with a private message and characterizing the MMSE performance of linear Bayesian estimation schemes for Gaussian sources in additive noise channels are described as well. These applications support the usefulness of the developed mathematical results and  the versatility of the extremal entropy inequality.

The rest of this paper is organized as follows. The extremal entropy inequality without a covariance constraint and its alternative proof are presented in Section \ref{sec2}. The extremal entropy inequality and its first alternative proof, which are the main results of this paper, are described in Section \ref{sec3}. In Section \ref{sec4}, several novel applications of the EEI are introduced, including a second much simplified alternative proof of the EEI, to illustrate the usefulness and relevance of the developed mathematical framework and results. Finally, Section \ref{sec5}  concludes  this paper.

\subsection{Notations}
Throughout this paper, random vectors are denoted by capital letters such as $X$ and $Y$, matrices are represented by bold capital letters such as $\boldsymbol{\Sigma}$ and $\mathbf{R}$, and $n$ and $n$-by-$n$ denote the dimension (size) of a random vector and a matrix, respectively. All information theoretic quantities are represented by conventional notations. For example, $h(X)$ and $I(X;Y)$ stand for  differential entropy of a random vector $X$ and mutual information between  random vector $X$ and  random vector $Y$, respectively. Conditional entropy and conditional mutual information are denoted as $h(X|Y)$ and $I(X;Y|Z)$, respectively. The notation $\preceq$ or $\succeq$ stands for  positive (semi)definite partial ordering between matrices, i.e., $\boldsymbol{\Sigma}_1 \preceq \boldsymbol{\Sigma}_2$ means $\boldsymbol{\Sigma}_2-\boldsymbol{\Sigma}_1$ is a positive semidefinite matrix \cite{Matrix:Horn}. In this paper, a positive definite matrix means a strictly positive definite matrix, and $\nabla_{\boldsymbol{\Sigma}}$ stands for the Jacobian matrix with respect to $\boldsymbol{\Sigma}$. The matrix $\mathbf{I}$ denotes an $n$-by-$n$ identity matrix, and the matrix $\mathbf{0}$ stands for  an $n$-by-$n$ zero matrix. Notation $\mathbb{E}[\cdot]$ denotes an expectation with respect to all random vectors inside $[\cdot]$, and $M_X(S)$ and $M_{X|Y}(S)$   stand for the  moment generating functions of  random vector $X$ and  random vector $X$ given $Y$, respectively. For simplicity, $\log$ denotes the natural logarithm.

\section{Entropy Power Inequality}\label{sec2}
Since the extremal entropy inequality is similar to the classical entropy power inequality, we first investigate a relationship between the EEI and the EPI. Without a covariance constraint, the EEI  is equivalent to the EPI as shown in Theorem \ref{thm1}.

\begin{theorem}\label{thm1}
For an arbitrary random vector $X$ with a covariance matrix $\boldsymbol{\Sigma}_X$ and a Gaussian random vector $W_G$ with a covariance matrix $\boldsymbol{\Sigma}_W$, there exists a Gaussian random vector $\tilde{X}_G$ which satisfies the following inequality:
\begin{eqnarray}\label{thm1e1_1}
    h(X) - \mu h(X+W_G) & \leq & h(\tilde{X}_G) - \mu h(\tilde{X}_G+W_G),
\end{eqnarray}
where the constant $\mu \geq 1$, all random vectors are independent of each other, $\boldsymbol{\Sigma}_W$ is a positive definite matrix, and $\tilde{X}_G$ is a Gaussian random vector which satisfies the following:
\begin{enumerate}
\item The covariance matrix of $\tilde{X}_G$ is represented by $\boldsymbol{\Sigma}_{\tilde{X}}$, and it is proportional to $\boldsymbol{\Sigma}_W$.
\item The differential entropy of $\tilde{X}_{G}$, $h(\tilde{X}_G)$, is equal to the differential entropy of $X$, $h(X)$.
\end{enumerate}
In addition, the inequality (\ref{thm1e1_1}) is equivalent to the EPI.
\end{theorem}
\begin{proof}
\begin{lemma}[Entropy Power Inequality \cite{EPI:Rioul}, \cite{Infor:Cover}]\label{lem1}
For independent random vectors $X_1$ and $X_2$,
\begin{eqnarray}
\label{lem1e1_1}h(X_1+X_2) &\geq& h(\tilde{X}_{G_1}+\tilde{X}_{G_2}),
\end{eqnarray}
where $\tilde{X}_{G_1}$ and $\tilde{X}_{G_2}$ are independent Gaussian random vectors, $h(\tilde{X}_{G_1})=h(X_1)$ and $h(\tilde{X}_{G_2})=h(X_2)$, and the covariance matrices of $\tilde{X}_{G_1}$ and $\tilde{X}_{G_2}$ are proportional.
\end{lemma}

Using Lemma \ref{lem1}, the following relations are obtained:
\begin{eqnarray}\label{thm1e2_1}
h(X) & = & h(\tilde{X}_G),\nonumber\\
h(X+W_G) & \geq & h(\tilde{X}_G+W_G),
\end{eqnarray}
where $\boldsymbol{\Sigma}_{\tilde{X}}$ is proportional to $\boldsymbol{\Sigma}_W$, i.e., $\boldsymbol{\Sigma}_{\tilde{X}}=\alpha \boldsymbol{\Sigma}_W$, and $\alpha$ is an appropriate constant which satisfies $h(X) = h(\tilde{X}_G)$. Therefore, the inequality (\ref{thm1e1_1}) is derived from Lemma 1, the EPI, and the proof of the inequality (\ref{thm1e1_1}) is completed.

If the inequality (\ref{thm1e1_1}) holds, $h(X+W_G) \geq h(\tilde{X}_G+W_G)$ since $h(X) = h(\tilde{X}_G)$, and $\boldsymbol{\Sigma}_{\tilde{X}}$ is proportional to $\boldsymbol{\Sigma}_{W}$. This is exactly the same as the EPI in Lemma \ref{lem1}. Therefore, the inequality (\ref{thm1e1_1}) is equivalent to the EPI.
\end{proof}

While Theorem \ref{thm1} shows a local upper bound, i.e., the upper bound is dependent on a random vector $X$, since $\alpha$ depends on the random vector $X$, we can also find a global upper bound as shown in Theorem \ref{thm2} and the reference \cite{Extremal:Liu}.

\begin{theorem}\label{thm2}
For an arbitrary random vector $X$ with a covariance matrix $\boldsymbol{\Sigma}_X$ and a Gaussian random vector $W_G$ with a covariance matrix $\boldsymbol{\Sigma}_W$, there exists a Gaussian random vector $X_G^*$ which satisfies the following inequalities:
\begin{eqnarray}
\label{thm2e1_1}    h(X) - \mu h(X+W_G) & \leq & h(X^*_G) - \mu h(X_G^*+W_G),\\
\label{thm2e1_2}    h(\tilde{X}_G) - \mu h(\tilde{X}_G+W_G) & \leq & h(X^*_G) - \mu h(X_G^*+W_G),
\end{eqnarray}
where the  constant $\mu > 1$, all random vectors are independent of each other, $\boldsymbol{\Sigma}_W$ is a positive definite matrix, $\tilde{X}_G$ stands for the  Gaussian random vector  defined in Theorem \ref{thm1}, and $X_G^*$ is a Gaussian random vector whose  covariance matrix $\boldsymbol{\Sigma}_{X^*}$ is represented by $(\mu-1)^{-1}\boldsymbol{\Sigma}_W$.
\end{theorem}
\begin{proof}
The proof, here, is a little different from the proof in \cite{Extremal:Liu}. In our proof, we deal with both a local upper bound and a global upper bound while a global upper bound is directly calculated in \cite{Extremal:Liu}.

Define the function $f(\alpha)$ as follows:
\begin{eqnarray}
\label{thm2e2_1}f(\alpha) & = & h(\tilde{X}_G)-\mu h(\tilde{X}_G+W_G)\nonumber\\
& = & \frac{n}{2} \log 2\pi e \left|\alpha \boldsymbol{\Sigma}_W \right|^{\frac{1}{n}} -\frac{\mu n}{2} \log 2\pi e \left|\alpha \boldsymbol{\Sigma}_W + \boldsymbol{\Sigma}_W\right|^{\frac{1}{n}},
\end{eqnarray}
where $n$ denotes the dimension of a random vector, and $|\cdot|$  stands for the determinant of a matrix.

Since $f(\alpha)$ is unimodal, and
\begin{eqnarray}
\label{thm2e3_1} \frac{d}{d \alpha}f(\alpha)\biggr|_{\alpha=(\mu-1)^{-1}} & = & \frac{n}{2(\mu-1)^{-1}}-\frac{\mu n}{2((\mu-1)^{-1}+1)}\nonumber\\
& = & 0,\nonumber\\
\frac{d^2}{d^2 \alpha}f(\alpha)\biggr|_{\alpha=(\mu-1)^{-1}} & = & -\frac{n}{2(\mu-1)^{-2}}+\frac{\mu n}{2((\mu-1)^{-1}+1)^2}\nonumber\\
& < & 0,
\end{eqnarray}
$f(\alpha)$ is maximized when $\alpha=(\mu-1)^{-1}$.

Therefore, from Theorem \ref{thm1}, the following inequality is derived as
\begin{eqnarray}
\label{thm2e4_1}h(X) - \mu h(X+W_G) & \leq & h(\tilde{X}_G) - \mu h(\tilde{X}_G+W_G)\nonumber\\
& = & f(\alpha)\nonumber\\
& \leq & f((\mu-1)^{-1})\nonumber\\
& = & h(X^*_G) - \mu h(X_G^*+W_G).
\end{eqnarray}

The inequalities (\ref{thm2e4_1}) include inequalities (\ref{thm2e1_1}) and (\ref{thm2e1_2}), and the validity of inequalities (\ref{thm2e1_1}) and (\ref{thm2e1_2}) is  proved. The upper bound in (\ref{thm2e4_1}) is a global maximum while the upper bound derived in Theorem \ref{thm1} is a local maximum.
\begin{rem}
When $\mu=1$, the inequalities (\ref{thm2e1_1}) and (\ref{thm2e1_2}) are also satisfied. However, we cannot specify the covariance matrix of $X_G^*$ since $h(X^*_G) - \mu h(X_G^*+W_G)$ is increasing with respect to $\boldsymbol{\Sigma}_{X^*}$ and it can be infinitely large as $\boldsymbol{\Sigma}_{X^*}$ is increased. Therefore, we omit the case when $\mu=1$ in Theorem \ref{thm2}.
\end{rem}
\end{proof}

As shown in Theorems  \ref{thm1} and \ref{thm2}, for $\mu \geq 1$, $h(X) - \mu h(X+W_G)$ is maximized when  random vector $X$ is Gaussian. However, when a covariance constraint is added in the inequalities (\ref{thm1e1_1}), (\ref{thm2e1_1}) and (\ref{thm2e1_2}), we cannot prove whether a Gaussian random vector still maximizes $h(X) - \mu h(X+W_G)$ or not, based on the same methods as described in the proofs  of Theorems \ref{thm1} and \ref{thm2}, since the  covariance constraint may alter the proportionality relationship between the covariance matrices $\boldsymbol{\Sigma}_{X^*}$ and $\boldsymbol{\Sigma}_{W}$.

\section{The Extremal Entropy Inequality}\label{sec3}

In \cite{Extremal:Liu}, Liu and Viswanath  proved that a Gaussian random vector still maximizes $h(X) - \mu h(X+W_G)$ even when a covariance constraint is  considered. The inequality (\ref{thm2e1_1}) was formulated  as an optimization problem with a covariance constraint as follows:
\begin{eqnarray}\label{Exine_e1_1}
&&\max\limits_{p(X)}\hspace{3mm} h(X+W_G)-\mu h(X+V_G),\nonumber\\
&&\textrm{s.t. }\hspace{5mm} \boldsymbol{\Sigma}_X \preceq \mathbf{R},
\end{eqnarray}
where $W_G$ and $V_G$ are independent Gaussian random vectors with positive definite covariance matrices $\boldsymbol{\Sigma}_{W}$ and $\boldsymbol{\Sigma}_V$, respectively, all random vectors are independent of each other, and the maximization is done over the distribution of random vector $X$. Two proofs, a direct proof and a perturbation proof, are provided in \cite{Extremal:Liu}. Each proof approaches the problem in a different way but both proofs share an important common approach, namely the channel enhancement technique based on the KKT conditions, proposed originally in \cite{CapBroad:Shamai}.

Unlike the original proofs in \cite{Extremal:Liu}, we will prove Theorems  \ref{thm3} and \ref{thm4} without using the channel enhancement technique. Before we deal with the problem (\ref{Exine_e1_1}), we first consider a simpler case of it next.

\begin{theorem}\label{thm3}
For an arbitrary random vector $X$ with a covariance matrix $\boldsymbol{\Sigma}_X$, a Gaussian random vector $W_G$ with a covariance matrix $\boldsymbol{\Sigma}_W$, and a positive semi-definite matrix $\mathbf{R}$, there exists a Gaussian random vector $X_G^*$ with a covariance matrix $\boldsymbol{\Sigma}_{X^*}$ which satisfies the following inequality:
\begin{eqnarray}\label{thm3e1_1}
    h(X) - \mu h(X+W_G) & \leq & h(X_G^*) - \mu h(X_G^*+W_G),
\end{eqnarray}
where the  constant $\mu \geq 1$, all random vectors are independent of each other, $\boldsymbol{\Sigma}_W$ is a positive definite matrix, $\boldsymbol{\Sigma}_X \preceq \mathbf{R}$, $\boldsymbol{\Sigma}_{X^*} \preceq \mathbf{R}$.
\end{theorem}
\begin{proof}
When $\mathbf{R}$ is a positive definite but singular matrix, i.e., $|\mathbf{R}|=0$, the inequality (\ref{thm3e1_1}) and its covariance constraints are equivalently changed into
\begin{eqnarray}
h(\bar{X}) - \mu h(\bar{X}+\bar{W}_G) & \leq & h(\bar{X}_G^*) - \mu h(\bar{X}_G^*+\bar{W}_G),
\end{eqnarray}
where $\bar{X}$ is such that $\boldsymbol{\Sigma}_{\bar{X}} \preceq \mathbf{\bar{R}}$, $\boldsymbol{\Sigma}_{\bar{X}^*} \preceq \mathbf{\bar{R}}$, and $\mathbf{\bar{R}}$ is a positive definite matrix, as mentioned in \cite{Extremal:Liu}. When $\mu=1$, the inequality (\ref{thm3e1_1}) is easily proved by the Lemma \ref{lem2}, which will be presented later.

Therefore, without loss of generality, we assume that $\mu>1$ and $\mathbf{R}$ is a positive definite matrix. Then, the right-hand side (RHS) of the equation (\ref{thm3e1_1}) is upper-bounded by means of the following lemma.
\begin{lemma}[Worst Additive Noise \cite{EPI:Rioul}, \cite{Extremal:Liu}, \cite{WANoise:Cover}]\label{lem2}
For random vectors $X$, $X_G$, $\tilde{W}_G$, and $W'_G$,
\begin{eqnarray}\label{lem2e1_1}
I(X+\tilde{W}_G+W'_G;W'_G) & \geq & I(X_G+\tilde{W}_G+W'_G;W'_G),
\end{eqnarray}
where $X$ is an arbitrary random vector, $X_G$ is a Gaussian random vector with the covariance matrix identical to that of $X$, $\tilde{W}_G$ and $W'_G$ are Gaussian random vectors, and all random vectors are independent.
\end{lemma}

Based on Lemma \ref{lem2}, the following inequalities hold:
\begin{eqnarray}
\label{thm3e2_1}&& h(X+\tilde{W}_G+W'_G)-h(X+\tilde{W}_G+W'_G|W'_G) \geq h(X_G+\tilde{W}_G+W'_G)-h(X_G+\tilde{W}_G+W'_G|W'_G)\nonumber\\
\label{thm3e2_2}& \Longleftrightarrow & h(X+\tilde{W}_G+W'_G)-h(X+\tilde{W}_G) \geq h(X_G+\tilde{W}_G+W'_G)-h(X_G+\tilde{W}_G)\\
\label{thm3e2_3}& \Longleftrightarrow & h(X+\tilde{W}_G+W'_G) \geq h(X+\tilde{W}_G) + h(X_G+\tilde{W}_G+W'_G)-h(X_G+\tilde{W}_G),
\end{eqnarray}
where $\Longleftrightarrow$ denotes equivalence. Notice that the Gaussian random vector $W_G$ can be expressed as the sum  of two independent Gaussian random vectors $\tilde{W}_G$ and $W'_G$ whose covariance matrices satisfy:
\begin{eqnarray}\label{thm3e3_0}
\boldsymbol{\Sigma}_{W} & = & \boldsymbol{\Sigma}_{\tilde{W}} +\boldsymbol{\Sigma}_{W'},
\end{eqnarray}
where $\boldsymbol{\Sigma}_{W}$, $\boldsymbol{\Sigma}_{\tilde{W}}$, and $\boldsymbol{\Sigma}_{W'}$ are the covariance matrices of $W_G$, $\tilde{W}_G$, and $W'_G$, respectively.  Henceforth, the Gaussian random vector $W_G$ is represented as $W_G=\tilde{W}_G+W'_G$.

Based on (\ref{thm3e2_3}) and  (\ref{thm3e3_0}), the left-hand side (LHS) of the equation (\ref{thm3e1_1}) is upper-bounded as follows:
\begin{eqnarray}
\label{thm3e3_1}h(X) - \mu h(X+W_G) & = & h(X) - \mu h(X+\tilde{W}_G+W'_G)\\
\label{thm3e3_2}& \leq & h(X) - \mu \left(h(X+\tilde{W}_G) + h(X_G+\tilde{W}_G+W'_G)-h(X_G+\tilde{W}_G)\right)\\
\label{thm3e3_3}& = & h(X) - \mu h(X+\tilde{W}_G) + \mu\left(h(X_G+\tilde{W}_G)- h(X_G+\tilde{W}_G+W'_G)\right).
\end{eqnarray}
Using Theorem \ref{thm2}, if $(\mu-1)^{-1}\boldsymbol{\Sigma_{\tilde{W}}}\preceq \mathbf{R}$, the RHS of equation (\ref{thm3e3_3}) is upper-bounded as follows:
\begin{eqnarray}
\label{thm3e4_1}&& h(X) - \mu h(X+\tilde{W}_G) + \mu\left(h(X_G+\tilde{W}_G)- h(X_G+\tilde{W}_G+W'_G)\right)\\
\label{thm3e4_2}& \leq & h(X^*_G) - \mu h(X^*_G+\tilde{W}_G) + \mu\left(h(X_G+\tilde{W}_G)- h(X_G+\tilde{W}_G+W'_G)\right),
\end{eqnarray}
where $X^*_G$ is a Gaussian random vector whose covariance matrix $\boldsymbol{\Sigma}_{X^*}$ is defined as $(\mu-1)^{-1}\boldsymbol{\Sigma}_{\tilde{W}}$. Unlike Theorem \ref{thm2}, we additionally have to prove that there exists a random vector $X^*_G$ whose covariance matrix $\boldsymbol{\Sigma}_{X^*}$ satisfies
\begin{eqnarray}
\label{thm3e5_1}\boldsymbol{\Sigma}_{X^*} & = & (\mu-1)^{-1}\boldsymbol{\Sigma}_{\tilde{W}}\\
\label{thm3e5_2}& \preceq & \mathbf{R},
\end{eqnarray}
due to the covariance constraint. Since $\boldsymbol{\Sigma}_{X}\preceq \mathbf{R}$, we will prove there exists a random vector $X^*_G$ whose covariance matrix $\boldsymbol{\Sigma}_{X^*}$ satisfies
\begin{eqnarray}
\label{thm3e6_1}\boldsymbol{\Sigma}_{X^*} & = & (\mu-1)^{-1}\boldsymbol{\Sigma}_{\tilde{W}}\\
\label{thm3e6_2}& \preceq & \boldsymbol{\Sigma}_{X},
\end{eqnarray}
instead of proving (\ref{thm3e5_2}). 

Equation (\ref{thm3e4_2}) is further processed by making use of the following lemma.
\begin{lemma}[Data Processing Inequality \cite{Infor:Cover}]\label{lem3}
When three random vectors $Y_1$, $Y_2$, and $Y_3$ represent a Markov chain  $Y_1\rightarrow Y_2 \rightarrow Y_3$,
the following inequality is satisfied:
\begin{eqnarray}\label{lem3e1_1}
    I(Y_1;Y_3) \leq I(Y_1;Y_2).
\end{eqnarray}
The equality holds if and only if random vectors $Y_1$, $Y_2$, and $Y_3$ form the Markov chain: $Y_1\rightarrow Y_3 \rightarrow Y_2$.
\end{lemma}

If the inequality (\ref{thm3e6_2}) is satisfied, then we can form the Markov chain:
\begin{eqnarray}
\label{thm3e7_1}X'_G \rightarrow X'_G + X^*_G + \tilde{W}_G \rightarrow X'_G + X^*_G + \tilde{W}_G + W'_G,
\end{eqnarray}
where all random vectors are independent. Since a Gaussian random vector $X_G$ can be expressed as the summation of two independent Gaussian random vectors $X'_G$ and $X^*_G$ whose covariance matrices satisfy
\begin{eqnarray}
\boldsymbol{\Sigma}_X & = & \boldsymbol{\Sigma}_{X'}+\boldsymbol{\Sigma}_{X^*},
\end{eqnarray}
where $\boldsymbol{\Sigma}_X$, $\boldsymbol{\Sigma}_{X'}$, and $\boldsymbol{\Sigma}_{X^*}$ stand for covariance matrices of $X_G$, $X'_G$, and $X^*_G$, respectively, the Gaussian random vector $X_G$ will be  represented as $X_G=X'_G+X^*_G$.

Based on Lemma \ref{lem3}, we obtain
\begin{eqnarray}
\label{thm3e8_1} \hspace*{-3mm}&&\hspace*{-3mm} I(X'_G;X'_G + X^*_G + \tilde{W}_G + W'_G) \leq I(X'_G;X'_G + X^*_G + \tilde{W}_G )\\
\label{thm3e8_2}\hspace*{-3mm}&\Longleftrightarrow& \hspace*{-3mm}h(X'_G + X^*_G + \tilde{W}_G + W'_G)-h(X^*_G + \tilde{W}_G + W'_G) \leq h(X'_G + X^*_G + \tilde{W}_G)-h( X^*_G + \tilde{W}_G )\\
\label{thm3e8_3}\hspace*{-3mm}&\Longleftrightarrow& \hspace*{-3mm} h(X_G+ \tilde{W}_G + W'_G)-h(X^*_G + \tilde{W}_G + W'_G) \leq h(X_G + \tilde{W}_G)-h( X^*_G + \tilde{W}_G )\\
\label{thm3e8_4}\hspace*{-3mm}&\Longleftrightarrow& \hspace*{-3mm}h(X^*_G+\tilde{W}_G) - h(X_G+\tilde{W}_G) + h(X_G+\tilde{W}_G+W'_G) \leq h(X^*_G + \tilde{W}_G + W'_G).
\end{eqnarray}
The equivalence in (\ref{thm3e8_3}) is due to $X_G=X'_G+X^*_G$.

Even though we need an upper bound of the RHS term in equation (\ref{thm3e4_2}),  the equation (\ref{thm3e8_4}) generates a lower bound for the equation (\ref{thm3e4_2}) as follows:
\begin{eqnarray}
\label{thm3e9_1}&&h(X^*_G) - \mu h(X^*_G+\tilde{W}_G) + \mu\left(h(X_G+\tilde{W}_G)- h(X_G+\tilde{W}_G+W'_G)\right)\\
\label{thm3e9_2}& \geq & h(X^*_G) - \mu h(X^*_G+\tilde{W}_G+W'_G)\\
\label{thm3e9_3}& \geq & h(X^*_G) - \mu h(X^*_G+W_G).
\end{eqnarray}

However, if we can construct the following Markov chain:
\begin{eqnarray}\label{thm3e10_1}
X'_G \rightarrow X'_G + X^*_G + \tilde{W}_G + W'_G \rightarrow X'_G + X^*_G + \tilde{W}_G ,
\end{eqnarray}
and using Lemma \ref{lem3} again, it turns out that
\begin{eqnarray}\label{thm3e11_1}
I(X'_G;X'_G + X^*_G + \tilde{W}_G + W'_G) \geq I(X'_G;X'_G + X^*_G + \tilde{W}_G ),
\end{eqnarray}
and this inequality leads us to a tight upper bound.  Indeed,
\begin{eqnarray}
\label{thm3e12_1}\hspace*{-3mm}&&\hspace*{-3mm}I(X'_G;X'_G + X^*_G + \tilde{W}_G + W'_G) \geq I(X'_G;X'_G + X^*_G + \tilde{W}_G )\\
\label{thm3e12_2}\hspace*{-3mm}&\Longleftrightarrow&\hspace*{-3mm} h(X'_G + X^*_G + \tilde{W}_G + W'_G)-h(X^*_G + \tilde{W}_G + W'_G) \geq h(X'_G + X^*_G + \tilde{W}_G)-h( X^*_G + \tilde{W}_G )\\
\label{thm3e12_3}\hspace*{-3mm}&\Longleftrightarrow&\hspace*{-3mm} h(X_G+ \tilde{W}_G + W'_G)-h(X^*_G + \tilde{W}_G + W'_G) \geq h(X_G + \tilde{W}_G)-h( X^*_G + \tilde{W}_G )\\
\label{thm3e12_4}\hspace*{-3mm}&\Longleftrightarrow& \hspace*{-3mm} h(X^*_G+\tilde{W}_G) - h(X_G+\tilde{W}_G) + h(X_G+\tilde{W}_G+W'_G) \geq h(X^*_G + \tilde{W}_G + W'_G).
\end{eqnarray}
The equivalence in (\ref{thm3e12_3}) is due to $X_G=X'_G+X^*_G$.

Now using (\ref{thm3e12_4}),  the equations (\ref{thm3e4_1}) and (\ref{thm3e4_2}) are upper-bounded as follows:
\begin{eqnarray}
\label{thm3e13_1}&& h(X) - \mu h(X+\tilde{W}_G) + \mu\left(h(X_G+\tilde{W}_G)- h(X_G+\tilde{W}_G+W'_G)\right)\\
\label{thm3e13_2}& \leq & h(X^*_G) - \mu h(X^*_G+\tilde{W}_G) + \mu\left(h(X_G+\tilde{W}_G)- h(X_G+\tilde{W}_G+W'_G)\right)\\
\label{thm3e13_3}& \leq & h(X^*_G) - \mu h(X^*_G+\tilde{W}_G+W'_G)\\
\label{thm3e13_4}& = & h(X^*_G) - \mu h(X^*_G+W_G),
\end{eqnarray}
and this is exactly the same as the equation (\ref{thm3e9_3}). Therefore, the following equality is satisfied:
\begin{eqnarray}
&   & h(X^*_G) - \mu h(X^*_G+\tilde{W}_G) + \mu\left(h(X_G+\tilde{W}_G)- h(X_G+\tilde{W}_G+W'_G)\right)\nonumber\\
\label{thm3e13_4_1}
& = & h(X^*_G) - \mu h(X^*_G+\tilde{W}_G+W'_G),
\end{eqnarray}
due to  (\ref{thm3e9_3}) and (\ref{thm3e13_4}).
Now, we will prove that we can actually construct the  Markov chain (\ref{thm3e10_1}) using the following lemmas.
\begin{lemma}\label{lem4}
For independent random vectors $Y_1$ and $Y_2$, the following equality between moment generating functions (MGFs) is satisfied:
\begin{eqnarray}
M_{Y_1+Y_2}(S) & = & M_{Y_1}(S)M_{Y_2}(S),
\end{eqnarray}\label{lem4e1_1}
where $M_{Y}(S) = \mathbb{E} [e^{Y^{T}S}]$, $\mathbb{E}[ \cdot ]$ is an expectation, and superscript $T$ denotes the  transpose of a vector. For jointly Gaussian random vectors $Y_1$ and $Y_2$, this equality is a necessary and sufficient condition for the independence between $Y_1$ and $Y_2$.
\end{lemma}
\begin{lemma}\label{lem5}
For independent random vectors $Y_1$ and $Y_2$ given a random vector $Y_3$, the following equality is satisfied:
\begin{eqnarray}\label{lem5e1_1}
M_{Y_1+Y_2|Y_3}(S) & = & M_{Y_1|Y_3}(S)M_{Y_2|Y_3}(S).
\end{eqnarray}
\end{lemma}
\begin{lemma}\label{lem6}
For a Gaussian random vector $X$ with a mean $U_X$ and a covariance matrix $\boldsymbol{\Sigma}_X$, the MGF is expressed as
\begin{eqnarray}\label{lem6e1_1}
M_{X}(S) & = & \exp \left\{S^T U_X + \frac{1}{2} S^T \boldsymbol{\Sigma}_{X} S\right\}.
\end{eqnarray}
\end{lemma}

In the Markov chain (\ref{thm3e10_1}), since all random vectors are Gaussian (without loss of generality, they are assumed to have zero means), using Lemma \ref{lem6}, the following moment generating functions are presented in closed-form expression:
\begin{eqnarray}
\label{thm3e14_1} M_{Y_1|Y_3}(S) & = & \exp \left\{S^T \boldsymbol{\Sigma}_{Y_1} \boldsymbol{\Sigma}_{Y_3}^{-1} Y_3 + \frac{1}{2} S^T \left(\boldsymbol{\Sigma}_{Y_1}-\boldsymbol{\Sigma}_{Y_1} \boldsymbol{\Sigma}_{Y_3}^{-1} \boldsymbol{\Sigma}_{Y_1}\right)S\right\},\nonumber\\
M_{Y_2|Y_3}(S) & = & \exp \left\{S^T \boldsymbol{\Sigma}_{Y_2} \boldsymbol{\Sigma}_{Y_3}^{-1} Y_3 + \frac{1}{2} S^T \left(\boldsymbol{\Sigma}_{Y_2}-\boldsymbol{\Sigma}_{Y_2} \boldsymbol{\Sigma}_{Y_3}^{-1} \boldsymbol{\Sigma}_{Y_2}\right)S\right\},
\end{eqnarray}
where $Y_1=X'_G$, $Y_2=X'_G + X^*_G + \tilde{W}_G$, $Y_3=X'_G + X^*_G + \tilde{W}_G + W'_G$, and their covariance matrices are represented by $\boldsymbol{\Sigma}_{Y_1}$, $\boldsymbol{\Sigma}_{Y_2}$, and $\boldsymbol{\Sigma}_{Y_3}$, respectively. Since $\boldsymbol{\Sigma}_{\tilde{W}}+\boldsymbol{\Sigma}_{W'}$ is a positive definite matrix, there exists the inverse of $\boldsymbol{\Sigma}_{Y_3}$.

On the other hand, the MGF of $Y_1+Y_2$ given $Y_3$ is represented as
\begin{eqnarray}
&   & M_{Y_1+Y_2|Y_3}(S)\nonumber\\
& = & \exp \left\{S^T \left(\boldsymbol{\Sigma}_{Y_1}+\boldsymbol{\Sigma}_{Y_2}\right) \boldsymbol{\Sigma}_{Y_3}^{-1} Y_3 + \frac{1}{2} S^T \left(\boldsymbol{\Sigma}_{Y_1}-\boldsymbol{\Sigma}_{Y_1} \boldsymbol{\Sigma}_{Y_3}^{-1} \boldsymbol{\Sigma}_{Y_1}+\boldsymbol{\Sigma}_{Y_2}-\boldsymbol{\Sigma}_{Y_2} \boldsymbol{\Sigma}_{Y_3}^{-1} \boldsymbol{\Sigma}_{Y_2}\right)S\right\}\nonumber\\
\label{thm3e15_1} &   & \times \exp \left\{S^T \left(\boldsymbol{\Sigma}_{Y_1} - \boldsymbol{\Sigma}_{Y_2} \boldsymbol{\Sigma}_{Y_3}^{-1} \boldsymbol{\Sigma}_{Y_1} + \boldsymbol{\Sigma}_{Y_1} - \boldsymbol{\Sigma}_{Y_1} \boldsymbol{\Sigma}_{Y_3}^{-1} \boldsymbol{\Sigma}_{Y_2}\right)S\right\}\nonumber\\
\label{thm3e15_2}
& = & M_{Y_1|Y_3}(S) M_{Y_2|Y_3}(S)\exp \underbrace{\left\{S^T \left(\boldsymbol{\Sigma}_{Y_1} - \boldsymbol{\Sigma}_{Y_2} \boldsymbol{\Sigma}_{Y_3}^{-1} \boldsymbol{\Sigma}_{Y_1} + \boldsymbol{\Sigma}_{Y_1} - \boldsymbol{\Sigma}_{Y_1} \boldsymbol{\Sigma}_{Y_3}^{-1} \boldsymbol{\Sigma}_{Y_2}\right)S\right\}}_{(A)}.
\end{eqnarray}
If the term (A) in (\ref{thm3e15_2})  vanishes, $Y_1$ and $Y_2$ are independent given $Y_3$, and the Markov chain (\ref{thm3e10_1}) is obtained. Using Lemma 11, (1) in \cite{CapBroad:Shamai}, we define the covariance matrix $\boldsymbol{\Sigma}_{\tilde{W}}$ as
\begin{eqnarray}\label{thm3e16_1}
\boldsymbol{\Sigma}_{\tilde{W}} & = & \left(\left(\boldsymbol{\Sigma}_{X} + \boldsymbol{\Sigma}_{W}\right)^{-1}+\mathbf{L}\right)^{-1}-\boldsymbol{\Sigma}_{X},
\end{eqnarray}
where $\mathbf{L}\succeq \mathbf{0}$, and $\mathbf{0}$ denotes an $n$-by-$n$ zero matrix. The positive semi-definite matrix $\mathbf{L}$ must be chosen to satisfy
\begin{eqnarray}
\label{thm3e16_2} \boldsymbol{\Sigma}_{X^*} & \preceq & \boldsymbol{\Sigma}_{X},\\
\label{thm3e16_3} \mathbf{L}\boldsymbol{\Sigma}_{X'} & = & \boldsymbol{\Sigma}_{X'}\mathbf{L} = \mathbf{0},
\end{eqnarray}
where $\boldsymbol{\Sigma}_{X^*}=(\mu-1)^{-1} \boldsymbol{\Sigma}_{\tilde{W}}$, $\boldsymbol{\Sigma}_{X'}=\boldsymbol{\Sigma}_{X}-\boldsymbol{\Sigma}_{X^*}$, $\mathbf{L} \succeq \mathbf{0}$.  Lemma \ref{lem7} will prove that such a positive semi-definite matrix $\mathbf{L}$ exists.
\begin{lemma}\label{lem7}
There exists a positive semi-definite matrix  $\mathbf{L}$ which satisfies
\begin{eqnarray}
\boldsymbol{\Sigma}_{X^*} \preceq \boldsymbol{\Sigma}_{X},\quad \mathbf{L}\boldsymbol{\Sigma}_{X'} =\boldsymbol{\Sigma}_{X'} \mathbf{L} =\mathbf{0},
\end{eqnarray}
where $\boldsymbol{\Sigma}_{\tilde{W}} = \left(\left(\boldsymbol{\Sigma}_{X} + \boldsymbol{\Sigma}_{W}\right)^{-1}+\mathbf{L}\right)^{-1}-\boldsymbol{\Sigma}_{X}$, $\boldsymbol{\Sigma}_{X^*}=(\mu-1)^{-1} \boldsymbol{\Sigma}_{\tilde{W}}$, $\boldsymbol{\Sigma}_{X'}=\boldsymbol{\Sigma}_{X}-\boldsymbol{\Sigma}_{X^*}$,  and $\boldsymbol{\Sigma}_{X}$ and $\boldsymbol{\Sigma}_W$  stand for a positive semi-definite matrix and a positive definite matrix, respectively.
\begin{proof}
See Appendix \ref{appA}.
\begin{rem}
By directly using Lemma \ref{lem7} in (\ref{thm3e13_4_1}), one can prove Theorem \ref{thm3}. However, we prefer to include explicitly in the proof the step which exploits the equality condition in the data processing inequality and the moment generating function. This is due to the following reasons.  First, the included step shows how to come up with Lemma \ref{lem7}, and helps to understand the intuition behind the proposed proof. Second, the proposed step guarantees the fact that the optimal solutions must force  the factor $(A)$ in (\ref{thm3e15_2}) to be zero. In other words, the proposed step provides the necessary condition for the optimality.
\end{rem}
\end{proof}
\end{lemma}

The equation (\ref{thm3e16_1}) can be re-written as
\begin{eqnarray}
\label{thm3e17_1}&   & \boldsymbol{\Sigma}_X +\boldsymbol{\Sigma}_{\tilde{W}}= \left(\left(\boldsymbol{\Sigma}_X +\boldsymbol{\Sigma}_{W}\right)^{-1}+ \mathbf{L}\right)^{-1}\\
\label{thm3e17_2}&\Longleftrightarrow& \left(\boldsymbol{\Sigma}_X+\boldsymbol{\Sigma}_{\tilde{W}}\right)^{-1} = \left(\boldsymbol{\Sigma}_X +\boldsymbol{\Sigma}_{W}\right)^{-1}+ \mathbf{L}.
\end{eqnarray}
Since $\mathbf{L} \boldsymbol{\Sigma}_{X'} = \boldsymbol{\Sigma}_{X'} \mathbf{L} = \mathbf{0}$, by multiplying $\boldsymbol{\Sigma}_{X'}$ to both sides of the equation (\ref{thm3e17_2}),
\begin{eqnarray}
\label{thm3e18_1} \left(\boldsymbol{\Sigma}_X+\boldsymbol{\Sigma}_{\tilde{W}}\right)^{-1}\boldsymbol{\Sigma}_{X'} & = & \left(\boldsymbol{\Sigma}_X +\boldsymbol{\Sigma}_{W}\right)^{-1}\boldsymbol{\Sigma}_{X'}+ \mathbf{L}\boldsymbol{\Sigma}_{X'}\nonumber\\
\label{thm3e18_2} & = & \left(\boldsymbol{\Sigma}_X +\boldsymbol{\Sigma}_{W}\right)^{-1}\boldsymbol{\Sigma}_{X'},
\end{eqnarray}
and
\begin{eqnarray}
\label{thm3e18_3} \boldsymbol{\Sigma}_{X'} \left(\boldsymbol{\Sigma}_X+\boldsymbol{\Sigma}_{\tilde{W}}\right)^{-1} & = & \boldsymbol{\Sigma}_{X'} \left(\boldsymbol{\Sigma}_X +\boldsymbol{\Sigma}_{W}\right)^{-1}+ \boldsymbol{\Sigma}_{X'}\mathbf{L}\nonumber\\
\label{thm3e18_4} & = & \boldsymbol{\Sigma}_{X'}\left(\boldsymbol{\Sigma}_X +\boldsymbol{\Sigma}_{W}\right)^{-1}.
\end{eqnarray}
Since random vectors $Y_1$, $Y_2$, and $Y_3$ are defined as $Y_1=X'_G$, $Y_2=X'_G + X^*_G + \tilde{W}_G$, and $Y_3=X'_G + X^*_G + \tilde{W}_G + W'_G$, respectively, and they are independent of each other, their covariance matrices are represented as

\begin{eqnarray}
\label{thm3e19_1} \boldsymbol{\Sigma}_{Y_1} & = & \boldsymbol{\Sigma}_{X'},\nonumber\\
\boldsymbol{\Sigma}_{Y_2} & = & \boldsymbol{\Sigma}_{X'}+\boldsymbol{\Sigma}_{X^*}+\boldsymbol{\Sigma}_{\tilde{W}},\nonumber\\
& = & \boldsymbol{\Sigma}_{X}+\boldsymbol{\Sigma}_{\tilde{W}},\nonumber\\
\boldsymbol{\Sigma}_{Y_3} & = & \boldsymbol{\Sigma}_{X'} + \boldsymbol{\Sigma}_{X^*} + \boldsymbol{\Sigma}_{\tilde{W}} +\boldsymbol{\Sigma}_{W'} \nonumber\\
& = & \boldsymbol{\Sigma}_{X'}+\boldsymbol{\Sigma}_{X^*}+\boldsymbol{\Sigma}_{W}\nonumber\\
& = & \boldsymbol{\Sigma}_{X}+\boldsymbol{\Sigma}_{W}.
\end{eqnarray}
From the equations (\ref{thm3e18_2}) and (\ref{thm3e19_1}),  it follows that
\begin{eqnarray}\label{thm3e20_1}
&&\boldsymbol{\Sigma}_{Y_1} - \boldsymbol{\Sigma}_{Y_2} \boldsymbol{\Sigma}_{Y_3}^{-1} \boldsymbol{\Sigma}_{Y_1}\nonumber\\
& = & \boldsymbol{\Sigma}_{X'} - \left(\boldsymbol{\Sigma}_{X'}+\boldsymbol{\Sigma}_{X^*}+\boldsymbol{\Sigma}_{\tilde{W}}\right) \left(\boldsymbol{\Sigma}_{X'}+\boldsymbol{\Sigma}_{X^*}+\boldsymbol{\Sigma}_{W}\right)^{-1}\boldsymbol{\Sigma}_{X'}\nonumber\\
& = & \left(\boldsymbol{\Sigma}_{X'}+\boldsymbol{\Sigma}_{X^*}+\boldsymbol{\Sigma}_{\tilde{W}}\right) \left(\left(\boldsymbol{\Sigma}_{X'}+\boldsymbol{\Sigma}_{X^*}+\boldsymbol{\Sigma}_{\tilde{W}}\right)^{-1}\boldsymbol{\Sigma}_{X'} -  \left(\boldsymbol{\Sigma}_{X'}+\boldsymbol{\Sigma}_{X^*}+\boldsymbol{\Sigma}_{W}\right)^{-1}\boldsymbol{\Sigma}_{X'}\right)\nonumber\\
& = & \mathbf{0},
\end{eqnarray}
and from the equations (\ref{thm3e18_4}) and (\ref{thm3e19_1}), one can infer that
\begin{eqnarray}\label{thm3e20_2}
&& \boldsymbol{\Sigma}_{Y_1} - \boldsymbol{\Sigma}_{Y_1} \boldsymbol{\Sigma}_{Y_3}^{-1} \boldsymbol{\Sigma}_{Y_2}\nonumber\\
& = & \boldsymbol{\Sigma}_{X'} - \boldsymbol{\Sigma}_{X'}\left(\boldsymbol{\Sigma}_{X'}+\boldsymbol{\Sigma}_{X^*}+\boldsymbol{\Sigma}_{W}\right)^{-1} \left(\boldsymbol{\Sigma}_{X'}+\boldsymbol{\Sigma}_{X^*}+\boldsymbol{\Sigma}_{\tilde{W}}\right)\nonumber\\
& = & \left(\boldsymbol{\Sigma}_{X'}\left(\boldsymbol{\Sigma}_{X'}+\boldsymbol{\Sigma}_{X^*}+\boldsymbol{\Sigma}_{\tilde{W}}\right)^{-1} - \boldsymbol{\Sigma}_{X'}\left(\boldsymbol{\Sigma}_{X'}+\boldsymbol{\Sigma}_{X^*}+\boldsymbol{\Sigma}_{W}\right)^{-1} \right)\left(\boldsymbol{\Sigma}_{X'}+\boldsymbol{\Sigma}_{X^*}+\boldsymbol{\Sigma}_{\tilde{W}}\right)\nonumber\\
& = & \mathbf{0}.
\end{eqnarray}
\end{proof}

The more general problem, originally proved in \cite{Extremal:Liu}, is now considered in Theorem \ref{thm4}.
\begin{theorem}\label{thm4}
For an arbitrary random vector $X$ with a covariance matrix $\boldsymbol{\Sigma}_X$, two independent random vectors $W_G$ and $V_G$ with covariance matrices $\boldsymbol{\Sigma}_W$ and $\boldsymbol{\Sigma}_V$, respectively, and a positive semi-definite matrix $\mathbf{R}$, there exists a Gaussian random vector $X^*_G$ with a covariance matrix $\boldsymbol{\Sigma}_{X^*}$ which satisfies the following inequality:
\begin{eqnarray}\label{thm4e1_1}
    h(X+W_G) - \mu h(X+V_G) \leq h(X^*_G+W_G) - \mu h(X^*_G+V_G),
\end{eqnarray}
where the  constant $\mu \geq 1$, all random vectors are independent of each other, $\boldsymbol{\Sigma}_W$ is a positive definite matrix, $\boldsymbol{\Sigma}_X \preceq \mathbf{R}$, $\boldsymbol{\Sigma}_{X^*} \preceq \mathbf{R}$.
\end{theorem}
\begin{proof}
Due to the same reason mentioned in the proof of Theorem \ref{thm3}, without loss of generality, we assume $\mu>1$ and $\mathbf{R}$ is a positive definite matrix. The proof is generally similar to the proof of Theorem \ref{thm3}. Using Lemma \ref{lem3}, the inequality (\ref{thm4e1_1}) can be expressed as
\begin{eqnarray}
\label{thm4e2_1} h(X+W_G) - \mu h(X+V_G) & \leq & h(X+\tilde{W}_G) - \mu h(X+V_G) + h(W_G) - h(\tilde{W}_G)\\
\label{thm4e2_2} & \leq & h(X_G^*+\tilde{W}_G) - \mu h(X_G^*+V_G) + h(W_G) - h(\tilde{W}_G)\\
\label{thm4e2_3} & = & h(X_G^*+W_G) - \mu h(X_G^*+V_G),
\end{eqnarray}
where $\tilde{W}_G$ is chosen to be a Gaussian random vector whose  covariance matrix, $\boldsymbol{\Sigma}_{\tilde{W}}$, satisfies
\begin{eqnarray}
\label{thm4e3_1} \boldsymbol{\Sigma}_{\tilde{W}} & \preceq & \boldsymbol{\Sigma}_{W},\\
\label{thm4e3_2} \boldsymbol{\Sigma}_{\tilde{W}} & \preceq & \left(\mu-1\right)^{-1}\boldsymbol{\Sigma}_{\tilde{V}},
\end{eqnarray}
where $\boldsymbol{\Sigma}_{\tilde{V}}$ is the covariance matrix of the  Gaussian random vector $\tilde{V}_G$, $V'_G$ is a Gaussian random vector with  covariance matrix $\boldsymbol{\Sigma}_{V'}$,  $V_G = \tilde{W}_G + \tilde{V}_G + V'_G$, and $\tilde{W}_G$, $\tilde{V}_G$, and $V'_G$ are independent of one another.

The inequality in (\ref{thm4e2_1}) is due to Lemma \ref{lem3}, the inequality (\ref{thm4e2_2}) is due to Theorem \ref{thm3}, and the equality (\ref{thm4e2_3}) will be proved using the equality condition in Lemma \ref{lem3}. We will also prove that there exists a Gaussian random vector $\tilde{W}_G$ which satisfies the equations (\ref{thm4e3_1}) and (\ref{thm4e3_2}) by proving later Lemma \ref{lem8}.

To satisfy the equality in the equation (\ref{thm4e2_3}), the equality condition in Lemma \ref{lem3} must be satisfied, and the following two Markov chains are  formed:
\begin{enumerate}
\item \begin{eqnarray}\label{thm4e4_1}
X_G^* \rightarrow X_G^*+\tilde{W}_G \rightarrow X_G^*+\tilde{W}_G+W'_G,
\end{eqnarray}
\item \begin{eqnarray}\label{thm4e4_2}
X_G^* \rightarrow X_G^*+\tilde{W}_G+W'_G \rightarrow X_G^*+\tilde{W}_G,
\end{eqnarray}
\end{enumerate}
where all random vectors are normally distributed, $\tilde{W}_G$ and $W'_G$ are independent of each other, $W_G=\tilde{W}_G+W'_G$, and $X_G^*$ is independent of other random vectors.

The Markov chain (\ref{thm4e4_1}) is naturally formed since $X_G^*$, $\tilde{W}_G$, and $W'_G$ are independent Gaussian random vectors. The validity of the Markov chain (\ref{thm4e4_2}) is proved using the concept of moment generating function. In the Markov chain (\ref{thm4e4_2}), since all random vectors are Gaussian (without loss of generality, they are assumed to have zero means), using Lemma \ref{lem6}, the following moment generating functions are expressed in closed-form:
\begin{eqnarray}
\label{thm4e5_1} M_{Y_1|Y_3}(S) & = & \exp \left\{S^T \boldsymbol{\Sigma}_{Y_1} \boldsymbol{\Sigma}_{Y_3}^{-1} Y_3 + \frac{1}{2} S^T \left(\boldsymbol{\Sigma}_{Y_1}-\boldsymbol{\Sigma}_{Y_1} \boldsymbol{\Sigma}_{Y_3}^{-1} \boldsymbol{\Sigma}_{Y_1}\right)S\right\},\nonumber\\
\label{thm4e5_2} M_{Y_2|Y_3}(S) & = & \exp \left\{S^T \boldsymbol{\Sigma}_{Y_2} \boldsymbol{\Sigma}_{Y_3}^{-1} Y_3 + \frac{1}{2} S^T \left(\boldsymbol{\Sigma}_{Y_2}-\boldsymbol{\Sigma}_{Y_2} \boldsymbol{\Sigma}_{Y_3}^{-1} \boldsymbol{\Sigma}_{Y_2}\right)S\right\},
\end{eqnarray}
where $Y_1=X_G^*$, $Y_2=X_G^*+\tilde{W}_G$, $Y_3=X_G^*+\tilde{W}_G+W'_G$, and their covariance matrices are represented by $\boldsymbol{\Sigma}_{Y_1}$, $\boldsymbol{\Sigma}_{Y_2}$, and $\boldsymbol{\Sigma}_{Y_3}$, respectively. Since $\boldsymbol{\Sigma}_{W}$ is a positive definite matrix, there always exists the inverse of $\boldsymbol{\Sigma}_{Y_3}$.

On the other hand, the MGF of $Y_1+Y_2$ given $Y_3$ is represented as
\begin{eqnarray}
&   & M_{Y_1+Y_2|Y_3}(S)\nonumber\\
\label{thm4e6_1} & = & \exp \left\{S^T \left(\boldsymbol{\Sigma}_{Y_1}+\boldsymbol{\Sigma}_{Y_2}\right) \boldsymbol{\Sigma}_{Y_3}^{-1} Y_3 + \frac{1}{2} S^T \left(\boldsymbol{\Sigma}_{Y_1}-\boldsymbol{\Sigma}_{Y_1} \boldsymbol{\Sigma}_{Y_3}^{-1} \boldsymbol{\Sigma}_{Y_1}+\boldsymbol{\Sigma}_{Y_2}-\boldsymbol{\Sigma}_{Y_2} \boldsymbol{\Sigma}_{Y_3}^{-1} \boldsymbol{\Sigma}_{Y_2}\right)S\right\}\nonumber\\
&   & \times \exp \left\{S^T \left(\boldsymbol{\Sigma}_{Y_1} - \boldsymbol{\Sigma}_{Y_2} \boldsymbol{\Sigma}_{Y_3}^{-1} \boldsymbol{\Sigma}_{Y_1} + \boldsymbol{\Sigma}_{Y_1} - \boldsymbol{\Sigma}_{Y_1} \boldsymbol{\Sigma}_{Y_3}^{-1} \boldsymbol{\Sigma}_{Y_2}\right)S\right\}\nonumber\\
\label{thm4e6_2} & = & M_{Y_1|Y_3}(S) M_{Y_2|Y_3}(S)\exp \underbrace{\left\{S^T \left(\boldsymbol{\Sigma}_{Y_1} - \boldsymbol{\Sigma}_{Y_2} \boldsymbol{\Sigma}_{Y_3}^{-1} \boldsymbol{\Sigma}_{Y_1} + \boldsymbol{\Sigma}_{Y_1} - \boldsymbol{\Sigma}_{Y_1} \boldsymbol{\Sigma}_{Y_3}^{-1} \boldsymbol{\Sigma}_{Y_2}\right)S\right\}}_{(B)}.
\end{eqnarray}
If the factor $(B)$ in (\ref{thm4e6_2})  vanishes, $Y_1$ and $Y_2$ are independent given $Y_3$, and the Markov chain (\ref{thm4e4_2}) is obtained. Using Lemma 11, (1) in \cite{CapBroad:Shamai}, we define a covariance matrix $\boldsymbol{\Sigma}_{\tilde{W}}$ as follows:
\begin{eqnarray}\label{thm4e7_1}
\boldsymbol{\Sigma}_{\tilde{W}} & = & \left(\boldsymbol{\Sigma}_{W}^{-1}+\mathbf{K}\right)^{-1},
\end{eqnarray}
where $\mathbf{K} \succeq 0$, $\mathbf{K} \boldsymbol{\Sigma}_{X^*} = \boldsymbol{\Sigma}_{X^*}\mathbf{K} =\mathbf{0}$, and $\mathbf{0}$ denotes an $n$-by-$n$ zero matrix. Then, there exists a positive semi-definite matrix $\mathbf{K}$ which satisfies
\begin{eqnarray}\label{thm4e8_1}
\boldsymbol{\Sigma}_{\tilde{W}} & \preceq & \left(\mu-1\right)^{-1} \boldsymbol{\Sigma}_{\tilde{V}},\\
\mathbf{K} \boldsymbol{\Sigma}_{X^*} & = & \boldsymbol{\Sigma}_{X^*} \mathbf{K}   = \mathbf{0},\\
\end{eqnarray}
where $\boldsymbol{\Sigma}_{X^*}=(\mu-1)^{-1}\boldsymbol{\Sigma}_{\tilde{V}} - \boldsymbol{\Sigma}_{\tilde{W}}$, and $\boldsymbol{\Sigma}_{\tilde{V}}$ is a positive semi-definite matrix, which satisfies the following condition: $\boldsymbol{\Sigma}_{\tilde{V}} \preceq \boldsymbol{\Sigma}_{V}$. The existence of matrix $\mathbf{K}$
  is proved by the following lemma.
\begin{lemma}\label{lem8}
There always exists a positive semi-definite matrix $\mathbf{K}$ which satisfies
\begin{eqnarray}
\label{lem8e1_1} \boldsymbol{\Sigma}_{\tilde{W}} & \preceq & \left(\mu-1\right)^{-1} \boldsymbol{\Sigma}_{\tilde{V}},\\
\label{lem8e1_2} \mathbf{K}\boldsymbol{\Sigma}_{X^*} & = & \boldsymbol{\Sigma}_{X^*} \mathbf{K} = \mathbf{0},
\end{eqnarray}
where $\boldsymbol{\Sigma}_{X^*}=(\mu-1)^{-1}\boldsymbol{\Sigma}_{\tilde{V}} - \boldsymbol{\Sigma}_{\tilde{W}}$, and $\boldsymbol{\Sigma}_{\tilde{W}} = \left(\boldsymbol{\Sigma}_{W}^{-1}+\mathbf{K}\right)^{-1}$.
\end{lemma}

Since $\boldsymbol{\Sigma}_{\tilde{W}}$ is defined as $\left(\boldsymbol{\Sigma}_{W}^{-1}+ \mathbf{K}\right)^{-1}$ in (\ref{thm4e7_1}), $\boldsymbol{\Sigma}_{\tilde{W}}$ satisfies
\begin{eqnarray}\label{thm4e9_1}
\left(\boldsymbol{\Sigma}_{X^*} +\boldsymbol{\Sigma}_{\tilde{W}}\right)^{-1} & = & \left(\boldsymbol{\Sigma}_{X^*} +\boldsymbol{\Sigma}_{W}\right)^{-1} + \mathbf{K},
\end{eqnarray}
based on Lemma 11, (1) in \cite{CapBroad:Shamai}.

Since $\mathbf{K} \boldsymbol{\Sigma}_{X^*} = \boldsymbol{\Sigma}_{X^*} \mathbf{K} = \mathbf{0}$, multiplying with  $\boldsymbol{\Sigma}_{X^*}$ both sides of the equation (\ref{thm4e9_1}), it follows that:
\begin{eqnarray}
\label{thm4e10_1} \left(\boldsymbol{\Sigma}_{X^*} +\boldsymbol{\Sigma}_{\tilde{W}}\right)^{-1}\boldsymbol{\Sigma}_{X^*} & = & \left(\boldsymbol{\Sigma}_{X^*} +\boldsymbol{\Sigma}_{W}\right)^{-1}\boldsymbol{\Sigma}_{X^*} + \mathbf{K}\boldsymbol{\Sigma}_{X^*}\nonumber\\
\label{thm4e10_2} & = & \left(\boldsymbol{\Sigma}_{X^*} +\boldsymbol{\Sigma}_{W}\right)^{-1}\boldsymbol{\Sigma}_{X^*},
\end{eqnarray}
and
\begin{eqnarray}
\label{thm4e10_3}\boldsymbol{\Sigma}_{X^*} \left(\boldsymbol{\Sigma}_{X^*} +\boldsymbol{\Sigma}_{\tilde{W}}\right)^{-1} & = & \boldsymbol{\Sigma}_{X^*} \left(\boldsymbol{\Sigma}_{X^*} +\boldsymbol{\Sigma}_{W}\right)^{-1} + \boldsymbol{\Sigma}_{X^*}\mathbf{K}\nonumber\\
\label{thm4e10_4} & = & \boldsymbol{\Sigma}_{X^*}\left(\boldsymbol{\Sigma}_{X^*} +\boldsymbol{\Sigma}_{W}\right)^{-1}.
\end{eqnarray}

Random vectors $Y_1$, $Y_2$, and $Y_3$ are defined as $Y_1=X_G^*$, $Y_2=X_G^*+\tilde{W}_G$, and $Y_3=X_G^*+\tilde{W}_G+W'_G$, respectively, and $X_G^*$, $\tilde{W}_G$, and $W'_G$ are independent of each other. Therefore, their covariance matrices are represented as
\begin{eqnarray}
\label{thm4e11_1} \boldsymbol{\Sigma}_{Y_1} & = & \boldsymbol{\Sigma}_{X^*},\nonumber\\
\label{thm4e11_2} \boldsymbol{\Sigma}_{Y_2} & = & \boldsymbol{\Sigma}_{X^*}+\boldsymbol{\Sigma}_{\tilde{W}},\nonumber\\
\label{thm4e11_3} \boldsymbol{\Sigma}_{Y_3} & = & \boldsymbol{\Sigma}_{X^*}+\boldsymbol{\Sigma}_{\tilde{W}}+\boldsymbol{\Sigma}_{W'}\nonumber\\
& = & \boldsymbol{\Sigma}_{X^*}+\boldsymbol{\Sigma}_{W}.
\end{eqnarray}
From  (\ref{thm4e10_2}) and (\ref{thm4e11_3}),  one can infer that
\begin{eqnarray}
\label{thm4e12_1} \boldsymbol{\Sigma}_{Y_1} - \boldsymbol{\Sigma}_{Y_2} \boldsymbol{\Sigma}_{Y_3}^{-1} \boldsymbol{\Sigma}_{Y_1} & = & \boldsymbol{\Sigma}_{X^*} - \left(\boldsymbol{\Sigma}_{X^*}+\boldsymbol{\Sigma}_{\tilde{W}}\right) \left(\boldsymbol{\Sigma}_{X^*}+\boldsymbol{\Sigma}_{W}\right)^{-1}\boldsymbol{\Sigma}_{X^*}\nonumber\\
& = & \left(\boldsymbol{\Sigma}_{X^*}+\boldsymbol{\Sigma}_{\tilde{W}}\right) \left(\left(\boldsymbol{\Sigma}_{X^*}+\boldsymbol{\Sigma}_{\tilde{W}}\right)^{-1}\boldsymbol{\Sigma}_{X^*} -  \left(\boldsymbol{\Sigma}_{X^*}+\boldsymbol{\Sigma}_{W}\right)^{-1}\boldsymbol{\Sigma}_{X^*}\right)\nonumber\\
& = & \mathbf{0},
\end{eqnarray}
and from  (\ref{thm4e10_4}) and (\ref{thm4e11_3}), it follows similarly that
\begin{eqnarray}
\label{thm4e12_2} \boldsymbol{\Sigma}_{Y_1} - \boldsymbol{\Sigma}_{Y_1} \boldsymbol{\Sigma}_{Y_3}^{-1} \boldsymbol{\Sigma}_{Y_2} & = & \boldsymbol{\Sigma}_{X^*} - \boldsymbol{\Sigma}_{X^*}\left(\boldsymbol{\Sigma}_{X^*}+\boldsymbol{\Sigma}_{W}\right)^{-1} \left(\boldsymbol{\Sigma}_{X^*}+\boldsymbol{\Sigma}_{\tilde{W}}\right)\nonumber\\
& = & \left(\boldsymbol{\Sigma}_{X^*}\left(\boldsymbol{\Sigma}_{X^*}+\boldsymbol{\Sigma}_{\tilde{W}}\right)^{-1} - \boldsymbol{\Sigma}_{X^*}\left(\boldsymbol{\Sigma}_{X^*}+\boldsymbol{\Sigma}_{W}\right)^{-1} \right)\left(\boldsymbol{\Sigma}_{X^*}+\boldsymbol{\Sigma}_{\tilde{W}}\right)\nonumber\\
& = & \mathbf{0}.
\end{eqnarray}
Since the inverse matrix of $\boldsymbol{\Sigma}_{\tilde{W}}$ exists, $(\boldsymbol{\Sigma}_{X^*}+\boldsymbol{\Sigma}_{\tilde{W}})^{-1}$ also exists.

Therefore, (B) in the equation (\ref{thm4e6_2}) is zero, and $M_{Y_1+Y_2|Y_3}(S)=M_{Y_1|Y_3}(S) M_{Y_2|Y_3}(S)$. It means $Y_1$ and $Y_3$ are independent given $Y_2$, i.e.,  $X_G^*$ and $X_G^*+\tilde{W}_G$ are independent given $X_G^*+\tilde{W}_G+W'_G$, and the Markov chain (\ref{thm4e4_2}) is valid. The equality in the equation (\ref{thm4e2_3}) is achieved by the above procedure, and the proof is completed.
\end{proof}

\section{Applications of the Proposed Results}\label{sec4}

The versatility  of the extremal entropy inequality was already illustrated  by means of  several applications in \cite{Extremal:Liu}. However, the original proofs of the extremal entropy inequality in \cite{Extremal:Liu} were based on the channel enhancement technique while one of those applications, the capacity of the vector Gaussian broadcast channel, had been already proved by the channel enhancement technique in \cite{CapBroad:Shamai}. Even though the EEI was adapted to prove the capacity of the vector Gaussian broadcast channel in \cite{Extremal:Liu}, it failed to show a novel perspective since the proof of the EEI relied on the channel enhancement technique  \cite{CapBroad:Shamai}. On the other hand, based on our proof, the extremal inequality shows not only its  usefulness  but also a novel perspective to prove the capacity of the vector Gaussian broadcast channel.

To illustrate the usefulness of the proposed mathematical framework for proving the EEI,  this section proposes three  additional applications for the mathematical results presented in Section \ref{sec3}.  First, an alternative much  simplified approach for proving the EEI is provided. Second, finding the optimal solution of a   broadcasting channel with a private message and characterizing the MMSE performance of a linear Bayesian estimator for a Gaussian source embedded in additive noise are presented as additional applications of the proposed results.

\subsection{Another Novel Simplified Approach for Establishing the EEI}

Based on the proof  presented in previous section, one can come up with another more simplified proof for the EEI. This method relies partly on  calculus of variations techniques and partly on the results established in the previous section. However, this novel framework  is very general and can be further adapted to proving many other information theoretic inequalities. In this regard, a companion paper  was submitted for publication  \cite{erchin}. The proposed simplified proof of the EEI runs as follows.

First, select a Gaussian random vector $\tilde{{W}}_{G}$ whose covariance matrix $\boldsymbol{\Sigma}_{ \tilde{W}}$ satisfies $\boldsymbol{\Sigma}_{\tilde{W}} \preceq \boldsymbol{\Sigma}_{W}$ and $\boldsymbol{\Sigma}_{\tilde{W}}\preceq \boldsymbol{\Sigma}_{V}$. Since the Gaussian random vectors ${V}_{G}$ and ${W}_{G}$ can be represented as the sum of two independent random vectors $\tilde{{W}}_{ G}$ and $\hat{{V}}_{G}$, and as the sum of  two independent random vectors $\tilde{{W}}_{G}$ and $\hat{{W}}_{G}$, respectively, the LHS of the equation (\ref{thm4e1_1}) is expressed  as follows:
\begin{eqnarray}
\label{EI_eq51_1}
&& h({X}+{W}_{G}) - \mu h({X} + {V}_{G})  \nonumber\\
& \leq &  h({X}+\tilde{{W}}_{G}) - \mu h({X}+{V}_{G})  + h({W}_{G}) - h(\tilde{{W}}_{G})\nonumber\\
\label{EI_eq51_2}
& = & h({X}+\tilde{{W}}_{G})  - \mu h({X}+\tilde{{W}}_{G}+\hat{{V}}_{G}) + h(\tilde{{W}}_{G}+\hat{{W}}_{G}) - h(\tilde{{W}}_{G}) .
\end{eqnarray}

Since the equation will be maximized over $f_{ X}(\mathbf{x})$, the last two terms in (\ref{EI_eq51_2}) are ignored,  and  by defining the new random vectors ${Y}$ and $\hat{{X}}$ as ${X}+\tilde{{W}}_{G}+\hat{{V}}_{ G}$ and ${X}+\tilde{{W}}_{ G}$, respectively, the inequality in (\ref{thm4e1_1}) is equivalently expressed as the following variational problem:
\begin{eqnarray}
\label{EI_eq52_1}
&&\max_{f_{ \hat{X}}, f_{ Y}} \quad  h(\hat{{X}}) - \mu h({Y}) + \mu\left(\mu-1\right)  h(\hat{{V}}_{ G})\\
&&\text{s. t. }\quad \int\int f_{ \hat{X}}(\mathbf{x}) f_{ \hat{V}}(\mathbf{y}-\mathbf{x})d\mathbf{x}d\mathbf{y}-1 = 0, \nonumber\\
&&\hspace{11mm}\int\int f_{ \hat{X}}(\mathbf{x}) f_{ \hat{V}}(\mathbf{y}-\mathbf{x}) \mathbf{x}\mathbf{x}^{ T} d\mathbf{x}d\mathbf{y} -\boldsymbol{\Sigma}_{ \hat{X}} \preceq \mathbf{0},\nonumber\\
&&\hspace{11mm}\int\int f_{ \hat{X}}(\mathbf{x}) f_{ \hat{V}}(\mathbf{y}-\mathbf{x}) \mathbf{y}\mathbf{y}^{ T} d\mathbf{x}d\mathbf{y} -{\boldsymbol{\Sigma}}_{ Y^*}= \mathbf{0},\nonumber\\
&&\hspace{11mm}\int\int f_{ \hat{X}}(\mathbf{x}) f_{ \hat{V}}(\mathbf{y}-\mathbf{x})\left( \mathbf{y}\mathbf{y}^{ T} - \mathbf{x}\mathbf{x}^{ T}  - \left(\mathbf{y}-\mathbf{x}\right)\left(\mathbf{y}-\mathbf{x}\right)^{ T} \right)d\mathbf{x}d\mathbf{y} = \mathbf{0},\nonumber\\
&&\hspace{11mm}-\int\int f_{ \hat{X}}(\mathbf{x}) f_{ \hat{V}}(\mathbf{y}-\mathbf{x}) \log  f_{ \hat{X}}(\mathbf{x})d\mathbf{x}d\mathbf{y}  = p_{ \hat{X}},\\
&&\hspace{11mm}f_{ Y}(\mathbf{y}) = \int f_{ \hat{X}}(\mathbf{x})f_{ \hat{V}}(\mathbf{y}-\mathbf{x}) d\mathbf{x},\nonumber
\end{eqnarray}
where $\mathbf{x}$ and $\mathbf{y}$ are vectors, $\hat{{X}}={X}+\tilde{{W}}_{ G}$, ${Y}=\hat{{X}}+\hat{{V}}_{ G}$, ${W}_{ G} = \tilde{{W}}_{ G} + \hat{{W}}_{ G}$, ${V}_{ G} = \tilde{{W}}_{ G} + \hat{{V}}_{ G}$, $\boldsymbol{\Sigma}_{ \hat{X}}=\boldsymbol{\Sigma} +\boldsymbol{\Sigma}_{ \tilde{W}} $, $\boldsymbol{\Sigma}_{ Y^*}= \boldsymbol{\Sigma}_{ X^*} + \boldsymbol{\Sigma}_{ V} $, and  $\boldsymbol{\Sigma}_{ X^*}$ is the covariance matrix of the optimal solution ${X}^*$.

Using Euler's equations, we can solve this variational problem, and the problem in (\ref{EI_eq52_1}) is maximized when both $\hat{X}$ and $Y$ are Gaussian random vectors  (as shown in Appendix \ref{appD}). The important thing to  remark here is that solving this variational problem requires only the calculus of variations, i.e., the proposed method does not require neither the classical EPI nor the worst additive noise lemma.

Next the following inequality is obtained:
\begin{eqnarray}
\label{EI_eq53_1}
&&  h({X}+\tilde{{W}}_{ G})  - \mu h({X}+\tilde{{W}}_{ G}+\hat{{V}}_{ G}) + h(\tilde{{W}}_{ G}+\hat{{W}}_{ G}) -h(\tilde{{W}}_{ G})\nonumber\\
&\leq&  h({X}_{ G}^*+\tilde{{W}}_{ G})  - \mu h({X}_{ G}^*+\tilde{{W}}_{ G}+\hat{{V}}_{ G}) + h(\tilde{{W}}_{ G} +\hat{{W}}_{ G}) -h(\tilde{{W}}_{ G}).
\end{eqnarray}
Based on Lemma \ref{lem8},  the RHS of the equation  (\ref{EI_eq53_1}) is expressed as
\begin{eqnarray}
\label{EI_eq54_1}
&&  h({X}_{ G}^*+\tilde{{W}}_{ G}) - \mu h({X}_{ G}^*+\tilde{{W}}_{ G}+\hat{{V}}_{ G})  + h(\tilde{{W}}_{ G}+\hat{{W}}_{ G}) - h(\tilde{{W}}_{ G})\nonumber\\
& = & h({X}_{ G}^*+ {W}_{ G}) - \mu h({X}_{ G}^*+\tilde{{W}}_{ G}+\hat{{V}}_{ G}) ,
\end{eqnarray}
and therefore, from the equations in (\ref{EI_eq51_1}), (\ref{EI_eq53_1}), and (\ref{EI_eq54_1}), we obtain the following EEI:
\begin{eqnarray}
\label{EI_eq57_1}
&& h({X}+{W}_{ G})- \mu h({X}+{V}_{ G})  \nonumber\\
& \leq & h({X}+\tilde{{W}}_{ G}) - \mu h({X}+{V}_{ G})  + h({W}_{ G}) - h(\tilde{{W}}_{ G})\nonumber\\
& = & h({X}+\tilde{{W}}_{ G}) - \mu h({X}+\tilde{{W}}_{ G}+\hat{{V}}_{ G}) + h(\tilde{{W}}_{ G}+\hat{{W}}_{ G}) - h(\tilde{{W}}_{ G})\nonumber\\
&\leq&  h({X}_{ G}^*+\tilde{{W}}_{ G}) - \mu h({X}_{ G}^*+\tilde{{W}}_{ G}+\hat{{V}}_{ G})  + h(\tilde{{W}}_{ G}+\hat{{W}}_{ G}) - h(\tilde{{W}}_{ G})\nonumber\\
& = & h({X}_{ G}^*+\tilde{{W}}_{ G})  - \mu h({X}_{ G}^*+\tilde{{W}}_{ G}+\hat{{V}}_{ G}) + h(\tilde{{W}}_{ G}+\hat{{W}}_{ G}) - h(\tilde{{W}}_{ G})\nonumber\\
& = & h({X}_{ G}^*+{W}_{ G}) - \mu h({X}_{ G}^*+{V}_{ G}) ,\nonumber
\end{eqnarray}
and the proof is completed.

\subsection{Broadcasting Channel with a Private Message}

 Consider the practical communication set-up depicted in Figure \ref{Fig1}, where a broadcasting channel with a private message is considered  from the perspective of the mean square-error (MSE) performance metric. The input-output relationship of this broadcast channel are governed by these equations:
\begin{eqnarray}
\label{new_app_eq1_1}
Y_1 & = & X + Z_{G_1},\nonumber\\
Y_2 & = & X + Z_{G_2},
\end{eqnarray}
where $Z_{G_1}$ and $Z_{G_2}$ are additive Gaussian noise vectors with zero means and covariance matrices $\boldsymbol{\Sigma}_{Z_{G_1}}$ and $\boldsymbol{\Sigma}_{Z_{G_2}}$, respectively. The covariance matrices:  $\boldsymbol{\Sigma}_{Z_{G_1}}$ and $\boldsymbol{\Sigma}_{Z_{G_2}}$  are assumed to be positive definite. Matrix   $\boldsymbol{\Sigma}_X$ denotes  the  covariance matrix of $X$, and $\mathbf{R}$ stands for a positive semi-definite matrix.  Random vectors  $X$, $Z_{G_1}$, and $Z_{G_2}$  are assumed independent of one another. Random vectors $Y_1$ and $Y_2$ denote the received signals at the receiver $1$ and the receiver $2$, respectively.

Assume that the message $X$ is expected to be decoded only at the receiver $1$, but the message $X$ can be decoded at both the receivers $1$ and $2$ if they receive the message $X$ and the MSEs are below a certain threshold $\mathrm{Tr}\{\mathbf{R}\}$, respectively. Therefore, the question here is whether or not we can find a random vector $X$ which guarantees the MSE at the receiver $1$ is below the threshold $\mathrm{Tr}\{\mathbf{R}\}$,  while the MSE at the receiver $2$ is above the threshold $\mathrm{Tr}\{\mathbf{R}\}$, i.e.,  the receiver $1$ can decode the message $X$, but the receiver $2$ cannot decode the message $X$. The notation $\mathrm{Tr}$ denotes the trace of a matrix. We are also interested in which distribution of $X$ is the most power efficient to maintain such a MSE performance at the two  receivers.

\begin{figure}[th]
\begin {center}
    \includegraphics[width=0.4\textwidth]{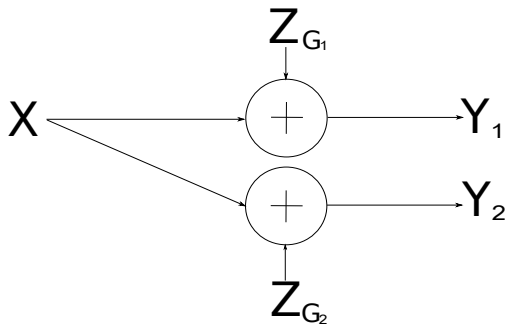}\\
  \caption{Gaussian broadcasting (wire-tap) channel}
\label{Fig1}
\end{center}
\end{figure}

To compare the MSE performance of the two receivers, we assume that both  receivers use  minimum mean-square error (MMSE) estimators. Since the minimum MSE estimator is optimal in the sense that it achieves the lowest MSE, this assumption is rational.

In summary, the goal of this problem is to find the optimal distribution $f_{X}(\mathbf{x})$ which satisfies the following problem:
\begin{eqnarray}
\label{new_app_eq1_3}
&&\min_{f_X(\mathbf{x})} \quad \boldsymbol{\Sigma}_{X} \\
&&\hspace{1mm} \mathrm{s.t.}\quad \quad \mathrm{Tr}\{\boldsymbol{\Sigma}_{X|Y_1}\} \leq \mathrm{Tr}\{\mathbf{R}\} \leq \mathrm{Tr}\{\boldsymbol{\Sigma}_{X|Y_2}\},\nonumber
\end{eqnarray}
where $\boldsymbol{\Sigma}_{X|Y_1}\hspace{-1mm}=\hspace{-1mm}\mathbb{E}\left[\left(X-\mathbb{E}\left[X|Y_1\right]\right)\left(X-\mathbb{E}\left[X|Y_1\right]\right)^T\right]$  and  $\boldsymbol{\Sigma}_{X|Y_2}\hspace{-1mm}=\hspace{-1mm}\mathbb{E}\left[\left(X-\mathbb{E}\left[X|Y_2\right]\right)\left(X-\mathbb{E}\left[X|Y_2\right]\right)^T\right]$.

The solution of the problem in (\ref{new_app_eq1_3}) can be obtained by the following procedure; first, define a new Gaussian random vector $\tilde{Z}_{G_1}$, which satisfies $\boldsymbol{\Sigma}_{\tilde{Z}_{G_1}} \preceq \boldsymbol{\Sigma}_{Z_{G_1}}$ and  $\boldsymbol{\Sigma}_{\tilde{Z}_{G_1}} \preceq \boldsymbol{\Sigma}_{Z_{G_2}}$, where $\boldsymbol{\Sigma}_{\tilde{Z}_{G_1}}$ stands for the  covariance matrix of $\tilde{Z}_{G_1}$. Second, find $\boldsymbol{\Sigma}_X$ which satisfies $\mathrm{Tr}\{\boldsymbol{\Sigma}_{X|Y_2}\} = \mathrm{Tr}\{\mathbf{R}\}$. Then, $\mathrm{Tr}\{\boldsymbol{\Sigma}_{X|\tilde{Y}_1}\} \leq \mathrm{Tr}\{\mathbf{R}\} = \mathrm{Tr}\{\boldsymbol{\Sigma}_{X|Y_2}\}$ since $\boldsymbol{\Sigma}_{X|\tilde{Y}_1} \preceq \boldsymbol{\Sigma}_{X|Y_2}$, where $\tilde{Y}_1=X+\tilde{Z}_{G_1}$, based on the data processing inequality for the covariance matrix \cite{EPI:Rioul}. Third, we will prove that there is a Gaussian $X_G^*$ with a covariance matrix $\boldsymbol{\Sigma}_{X_{G}^*}$, which satisfies $\boldsymbol{\Sigma}_{X_{G}^*|Y_2^*} = \boldsymbol{\Sigma}_{X|Y_2}$ and $\boldsymbol{\Sigma}_{X_{G}^*} \preceq \boldsymbol{\Sigma}_X$, where $Y_2^*=X_G^* + Z_{G_2}$. Finally, based on Lemma \ref{lem8}, we will show $\boldsymbol{\Sigma}_{X_{G}^*|Y_1^*}=\boldsymbol{\Sigma}_{X_{G}^*|\tilde{Y}_1^*}$, where $Y_1^*=X_G^*+Z_{G_1}$ and $\tilde{Y}_1^*=X_G^*+\tilde{Z}_{G_1}$. Since $\boldsymbol{\Sigma}_{X_G^*}$ is less than or equal to an arbitrary covariance matrix $\boldsymbol{\Sigma}_X$, which satisfies $\mathrm{Tr}\{\boldsymbol{\Sigma}_{X|Y_2}\} = \mathrm{Tr}\{\mathbf{R}\}$, the Gaussian random vector $X_G^*$ is the optimal solution (the details of the proof are deferred to  Appendix \ref{appC}).

Therefore, by choosing the message $X$ as a Gaussian random vector in (\ref{new_app_eq1_1}), we can securely transmit a private message, which is designed to arrive at the receiver $1$, in the most power efficient way.

\begin{rem}
This scenario can be interpreted as the secure transmission under a vector Gaussian wire-tap channel. In this case, the receiver $1$ is the legitimate receiver, and the receiver $2$ is the eavesdropper.
\end{rem}

\subsection{Bayesian Estimation of a Gaussian Source in Additive Noise Channel}

As shown in Figure \ref{Fig2}, the following additive noise channel is considered:
\begin{eqnarray}
\label{ANC_eq1_1}
Y = X_G + Z,
\end{eqnarray}
where $X_G$ is a Gaussian random vector with zero mean and covariance matrix $\boldsymbol{\Sigma}_{X}$, $Z$ denotes  an arbitrary random vector (noise) with zero mean and covariance matrix $\boldsymbol{\Sigma}_Z$, and $X_G$ and $Z$ are assumed independent of each other. We also assume that the covariance matrix of  additive noise $Z$ is upper-bounded, i.e., $\boldsymbol{\Sigma}_Z \preceq \mathbf{R}$, where $\mathbf{R}$ is a given positive semi-definite matrix.

\begin{figure}[th]
\begin {center}
    \includegraphics[width=0.4\textwidth]{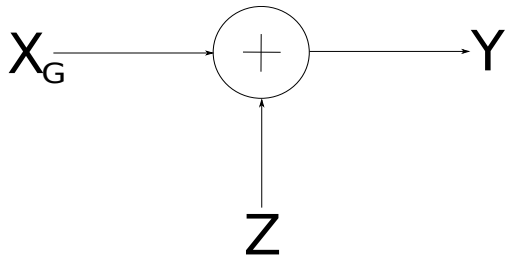}\\
  \caption{Additive Noise Channel}
\label{Fig2}
\end{center}
\end{figure}

 Using the channel model in (\ref{ANC_eq1_1}), we will next analyze the link between the channel input-output mutual information and the MMSE performance of a linear Bayesian estimator.  First, consider the following optimization problem:
\begin{eqnarray}
\label{ANC_eq2_1}
&& \min_{f_Z(\mathbf{Z})}\quad I(X_G; Y)\\
&&\mathrm{s.t.}\quad \quad \boldsymbol{\Sigma}_Z \preceq \mathbf{R}.\nonumber
\end{eqnarray}

The objective function in (\ref{ANC_eq2_1}) can be  expressed as
\begin{eqnarray}
\label{ANC_eq3_1}
I(X_G; Y) & = & I(X_G; X_G+Z)\nonumber\\
&=&h(X_G+Z) - h(Z),\nonumber
\end{eqnarray}
using the extremal inequality in Theorem \ref{thm3}, it follows  that the optimal solution of the optimization problem in (\ref{ANC_eq2_1}) is a multi-variate Gaussian density function, and the objective criterion can be expressed as
\begin{eqnarray}
\label{ANC_eq4_1}
I(X_G; X_G + Z_G^*) & = & h(X_G+Z_G^*) - h(Z_G^*)\nonumber\\
& = & -\frac{1}{2} \log \frac{\left|\boldsymbol{\Sigma}_X\right|}{\left|\boldsymbol{\Sigma}_Y\right|},
\end{eqnarray}
where $\boldsymbol{\Sigma}_{Y}$ is the covariance matrix of $Y$, $Y=X_G+Z_G^*$, and $Z_G^*$ is a Gaussian random vector with the covariance matrix $\mathbf{R}$. The right-hand side of  (\ref{ANC_eq4_1})  can be further expressed in terms of the MSE matrix of  the linear minimum MSE (LMMSE) estimator under the worst case scenario, i.e., the covariance matrix $\boldsymbol{\Sigma}_Z = \mathbf{R}$, as follows.  Given the channel (\ref{ANC_eq1_1}), the LMMSE estimator for $X$ takes the form:
\begin{eqnarray}
\hat{X} & = & \mathbb{E}\left[ X \right] + \boldsymbol{\Sigma}_X \left(\boldsymbol{\Sigma}_X +\mathbf{R}\right)^{-1} \left(Y - \mathbb{E}\left[X\right]\right)\nonumber\\
\label{ANC_eq5_1}
& = & \boldsymbol{\Sigma}_X \left(\boldsymbol{\Sigma}_X +\mathbf{R}\right)^{-1} Y,
\end{eqnarray}
where $\mathbb{E}[\cdot]$ denotes the expectation, and $\hat{X}$ stands for  the LMMSE estimator of $X$. The equality in (\ref{ANC_eq5_1}) is due to zero mean of $X$. Therefore, its MSE is expressed as
\begin{eqnarray}
\label{ANC_eq6_1}
\mathrm{LMMSE} & = & \boldsymbol{\Sigma}_X - \boldsymbol{\Sigma}_X \boldsymbol{\Sigma}_Y^{-1} \boldsymbol{\Sigma}_X \nonumber\\
& = &  \boldsymbol{\Sigma}_X  \boldsymbol{\Sigma}_Y^{-1}  \mathbf{R}.
\end{eqnarray}

Using  (\ref{ANC_eq6_1}), the equation in (\ref{ANC_eq4_1}) is expressed as
\begin{eqnarray}
\label{ANC_eq7_1}
I(X_G; X_G + Z_G^*) & = & -\frac{1}{2} \log \left|\mathrm{LMMSE}\right| + \frac{1}{2} \log \left|\mathbf{R}\right|,\nonumber\\
\end{eqnarray}
and it follows that
\begin{eqnarray}
\label{ANC_eq8_1}
I(X_G;X_G+Z) &\geq & I(X_G; X_G + Z_G^*)\nonumber\\
& = & -\frac{1}{2} \log \left|\mathrm{LMMSE}\right| + \frac{1}{2} \log \left|\mathbf{R}\right|.
\end{eqnarray}

Based on the equations in (\ref{ANC_eq8_1}), we can conclude the following facts.  First, when the additive noise is Gaussian, minimizing LMMSE is equivalent to maximizing the mutual information between the input and the output. Second, the worst case mutual information is expressed in terms of the LMMSE, while the mutual information in general is lower bounded by a function of the LMMSE. Finally, we observe  that the LMMSE estimator is, in general, sub-optimal since the mutual information between the input and the output is larger than the function of the LMMSE as shown in (\ref{ANC_eq8_1}).

\section{Conclusions}\label{sec5}
The main contributions of this paper are summarized as follows. In the first part of this paper, an alternative proof of the extremal entropy inequality is described in detail. The alternative proof is simpler, more direct, more explicit and more information-theoretic than the original proofs. The alternative proof is mainly based on the data processing inequality which enables  to by-pass the KKT conditions. Moreover, using properties of positive semi-definite matrices,  one can skip the step of  proving  the existence of the optimal solution which satisfies the KKT conditions, a step which is quite complicated to justify. This novel technique is based on a data processing inequality, and it is very unique and creative in respect that it presents a novel paradigm for lots of applications such as the capacity of the vector Gaussian broadcast channel and the secrecy capacity of the Gaussian wire-tap channel, which were proved commonly based on the channel enhancement technique \cite{Extremal:Liu}, \cite{CapBroad:Shamai}, \cite{SecrecyCap:Liu}, and \cite{SecCap:Ekrem}. Additional relevant applications in this regard include \cite{SecCap:Eldar}-\cite{sennur}. 
 In the second part of this paper, several additional important applications for the extremal entropy inequality are  presented. In this regard,  a second and even more simplified approach for establishing the extremal entropy inequality without using EPI or the worst data processing lemma is presented by exploiting the mathematical tools developed in the first part of this paper.   Two additional applications of the proposed mathematical results are presented  for the problem of determining the optimal solution for the broadcasting channel problem with a private message and in establishing  a mutual information-based performance bound for the mean square-error of a linear Bayesian estimator for a Gaussian source in an additive noise channel. One can observe that the last application presented can be sightly extended to non-Gaussian sources, a fact that suggests that the extremal entropy inequality (\ref{thm1e1_1}) might hold true even for non-Gaussian  $W_G$. However, establishing an extension of the EEI in this direction or other directions such as proving the EEI under a more general constraint (such as an upper and/or lower-bound constraint on the power spectral density of the random vector $X$ instead of the covariance matrix constraint) represent interesting open problems.  Finally, we would like to thank the reviewers for their constructive comments and bringing to our attention the reference
 \cite{nair} which presents a completely different approach for proving EEI. This  approach  relies on showing the optimality of Gaussian distribution  by exploiting  the factorization of concave envelopes. At the time of submitting our paper in August 2011, we were not aware of the parallel submission \cite{nair}.

\noindent

\appendices
\section{Proof of Lemma \ref{lem7}}\label{appA}

Proving $\boldsymbol{\Sigma}_{X^*} \preceq \boldsymbol{\Sigma}_X$ is equivalent to proving the following:
\begin{eqnarray}
\label{appendeq1_1} &    & \boldsymbol{\Sigma}_{X^*} \preceq \boldsymbol{\Sigma}_X\\
\label{appendeq1_2} &\Longleftrightarrow& \boldsymbol{\Sigma}_{\tilde{W}} \preceq \left(\mu-1\right)\boldsymbol{\Sigma}_X\\
\label{appendeq1_3} &\Longleftrightarrow& \left(\left(\boldsymbol{\Sigma}_{X}+\boldsymbol{\Sigma}_{W}\right)^{-1}+\mathbf{L}\right)^{-1}-\boldsymbol{\Sigma}_{X} \preceq \left(\mu-1\right)\boldsymbol{\Sigma}_X\\
\label{appendeq1_4} &\Longleftrightarrow& \left(\boldsymbol{\Sigma}_{X}+\boldsymbol{\Sigma}_{W}\right)^{-1}+\mathbf{L} \succeq \mu^{-1}\boldsymbol{\Sigma}_X^{-1}
\end{eqnarray}
Since there always exists a non-singular matrix which simultaneously diagonalizes two positive semi-definite matrices \cite{Matrix:Horn}, there exists a non-singular matrix $\mathbf{Q}$ which simultaneously diagonalizes both $\boldsymbol{\Sigma}_{X}$ and $\boldsymbol{\Sigma}_{W}$ as follows:
\begin{eqnarray}
\label{appendeq2_1}\mathbf{Q}^T \boldsymbol{\Sigma}_{X} \mathbf{Q} & = & \mathbf{I},\\
\label{appendeq2_2}\mathbf{Q}^T \boldsymbol{\Sigma}_{W} \mathbf{Q} & = & \mathbf{D}_W,\\
\end{eqnarray}
where $\mathbf{I}$ is an identity matrix, and $\mathbf{D}_W$ is a diagonal matrix. Since $\mathbf{Q}$ is a non-singular matrix, the inverse of $\mathbf{Q}$ always exists, and $\boldsymbol{\Sigma}_{X}$ and $\boldsymbol{\Sigma}_{W}$ are expressed as
\begin{eqnarray}
\label{appendeq3_1}\boldsymbol{\Sigma}_{X} & = & \mathbf{Q}^{-T}\mathbf{Q}^{-1},\\
\label{appendeq3_2}\boldsymbol{\Sigma}_{W} & = & \mathbf{Q}^{-T}\mathbf{D}_W \mathbf{Q}^{-1}.
\end{eqnarray}
If we define $\mathbf{D}_L$ as a diagonal matrix whose $i^{th}$ diagonal element is represented as $d_{L_i}$,  and which it is defined as
\begin{equation}
\label{appendeq4_1}d_{L_i} =  \left\{ \begin{array} {ll} 0 & \text{if }d_{W_i} \leq \mu-1 \\
\frac{d_{W_i}-\left(\mu-1\right)}{\mu\left(1+d_{W_i}\right)} & \text{if }d_{W_i} > \mu-1 \end{array} \right.
\end{equation}
where $d_{W_i}$ denotes the $i^{th}$ diagonal element of $\mathbf{D}_{W}$, and define $\mathbf{L}$ as
\begin{eqnarray}
\label{appendeq5_1}\mathbf{L} & = & \mathbf{Q} \mathbf{D}_L \mathbf{Q}^{T},
\end{eqnarray}
the equation (\ref{appendeq1_4}) is equivalent to
\begin{eqnarray}
\label{appendeq6_1}&&\left(\boldsymbol{\Sigma}_{X}+\boldsymbol{\Sigma}_{W}\right)^{-1}+\mathbf{L} \succeq \mu^{-1}\boldsymbol{\Sigma}_X^{-1}\\
\label{appendeq6_2}&\Longleftrightarrow& \left(\mathbf{Q}^{-T}\mathbf{Q}^{-1}+\mathbf{Q}^{-T}\mathbf{D}_W\mathbf{Q}^{-1}\right)^{-1}+\mathbf{Q}\mathbf{D}_L\mathbf{Q}^T \succeq \mu^{-1}\mathbf{Q}\mathbf{Q}^{T}\\
\label{appendeq6_3}&\Longleftrightarrow& \left(\mathbf{I}+\mathbf{D}_W\right)^{-1}+\mathbf{D}_L \succeq \mu^{-1}\mathbf{I}.
\end{eqnarray}
The equation (\ref{appendeq6_3}) always holds since $\mathbf{D}_L$ is defined as in (\ref{appendeq4_1}) and (\ref{appendeq5_1}) to satisfy (\ref{appendeq6_3}). Therefore, the inequality (\ref{appendeq1_1}) is also satisfied.

We know that $\boldsymbol{\Sigma}_{X'}$ is $\boldsymbol{\Sigma}_{X}-\boldsymbol{\Sigma}_{X^*}$. Since $\boldsymbol{\Sigma}_{X^*}=(\mu-1)^{-1}\boldsymbol{\Sigma}_{\tilde{W}}$, $\boldsymbol{\Sigma}_{X'}$ is expressed as $\boldsymbol{\Sigma}_{X}-(\mu-1)^{-1}\boldsymbol{\Sigma}_{\tilde{W}}$, and
\begin{eqnarray}
\label{appendeq7_1}
\boldsymbol{\Sigma}_{X'} \mathbf{L} & = & \left(\boldsymbol{\Sigma}_{X}-\left(\mu-1 \right)^{-1} \boldsymbol{\Sigma}_{\tilde{W}}\right) \mathbf{L},
\end{eqnarray}
and the equation (\ref{appendeq7_1}) is re-written as
\begin{eqnarray}
\label{appendeq8_1}\boldsymbol{\Sigma}_{X'} \mathbf{L}  \hspace{-3mm} & = &   \hspace{-3mm}\left(\boldsymbol{\Sigma}_{X}-\left(\mu-1 \right)^{-1} \boldsymbol{\Sigma}_{\tilde{W}}\right) \mathbf{L}\\
\label{appendeq8_2}
& = &  \hspace{-3mm} \left\{\mathbf{Q}^{-T}\mathbf{Q}^{-1}-\left(\mu-1 \right)^{-1} \left( \left(\left(\mathbf{Q}^{-T}\mathbf{Q}^{-1}+\mathbf{Q}^{-T}\mathbf{D}_{W}\mathbf{Q}^{-1}\right)^{-1}+\mathbf{Q} \mathbf{D}_L\mathbf{Q}^T\right)^{-1}- \mathbf{Q}^{-T}\mathbf{Q}^{-1}\right)\right\} \nonumber\\
&&\quad \quad \times \mathbf{Q}\mathbf{D}_L\mathbf{Q}^T\\
\label{appendeq8_3}& = &  \hspace{-3mm}\left(\mu-1 \right)^{-1} \mathbf{Q}^{-T}\left(\mu \mathbf{I}- \left(\left(\mathbf{I}+\mathbf{D}_{W}\right)^{-1} +  \mathbf{D}_L\right)^{-1}\right) \mathbf{D}_L\mathbf{Q}^T\\
\label{appendeq8_4}& = &  \hspace{-3mm} \mathbf{0}.
\end{eqnarray}
The equality (\ref{appendeq8_2}) is due to the equations (\ref{appendeq3_1}), (\ref{appendeq3_2}), and (\ref{appendeq5_1}), and the equality (\ref{appendeq8_4}) is due to (\ref{appendeq4_1}). Similarly,
\begin{eqnarray}
 \mathbf{L}  \boldsymbol{\Sigma}_{X'} \hspace{-3mm} & = &   \hspace{-3mm} \mathbf{L}\left(\boldsymbol{\Sigma}_{X}-\left(\mu-1 \right)^{-1} \boldsymbol{\Sigma}_{\tilde{W}}\right)\nonumber\\
& = & \hspace{-3mm} \mathbf{Q}\mathbf{D}_L\mathbf{Q}^T \nonumber\\
&  &  \hspace{-3mm} \times \left\{\mathbf{Q}^{-T}\mathbf{Q}^{-1}-\left(\mu-1 \right)^{-1} \left( \left(\left(\mathbf{Q}^{-T}\mathbf{Q}^{-1}+\mathbf{Q}^{-T}\mathbf{D}_{W}\mathbf{Q}^{-1}\right)^{-1}+\mathbf{Q} \mathbf{D}_L\mathbf{Q}^T\right)^{-1}- \mathbf{Q}^{-T}\mathbf{Q}^{-1}\right)\right\} \nonumber\\
& = &  \hspace{-3mm}\left(\mu-1 \right)^{-1} \mathbf{Q}\mathbf{D}_L \left(\mu \mathbf{I}- \left(\left(\mathbf{I}+\mathbf{D}_{W}\right)^{-1} +  \mathbf{D}_L\right)^{-1}\right)\mathbf{Q}^{-1} \nonumber\\
& = &  \hspace{-3mm} \mathbf{0}.\nonumber
\end{eqnarray}

Therefore, by defining $\boldsymbol{\Sigma}_{\tilde{W}}=((\boldsymbol{\Sigma}_{X}+\boldsymbol{\Sigma}_{W})^{-1}+\mathbf{L})^{-1}-\boldsymbol{\Sigma}_{X} $, we can make $\boldsymbol{\Sigma}_{\tilde{W}}$ satisfy
\begin{eqnarray}
\label{apendeq8_1}\boldsymbol{\Sigma}_{\tilde{W}} \preceq \left(\mu-1\right) \boldsymbol{\Sigma}_X,\hspace{5mm}
\label{apendeq8_2}\boldsymbol{\Sigma}_{X'} \mathbf{L} = \mathbf{L} \boldsymbol{\Sigma}_{X'}  = \mathbf{0},
\end{eqnarray}
and the proof is completed.
\begin{rem}
Since the optimization problem in \cite{Extremal:Liu} is generally nonconvex, the existence of optimal solution must be proved \cite{Extremal:Liu}, \cite{CapBroad:Shamai}, and this step is very complicated. However, in our proof, Lemmas \ref{lem7} and  \ref{lem8} serve as a substitute for this step  since we by-pass KKT-condition related parts using the data processing inequality. This makes the proposed proof much simpler.
\end{rem}

\section{Proof of Lemma \ref{lem8}}\label{appB}

Proving $\boldsymbol{\Sigma}_{\tilde{W}} \preceq \left(\mu-1\right)^{-1} \boldsymbol{\Sigma}_{\tilde{V}}$ is equivalent to proving the following:
\begin{eqnarray}
\label{apendeq1_1}&& \boldsymbol{\Sigma}_{\tilde{W}} \preceq \left(\mu-1\right)^{-1} \boldsymbol{\Sigma}_{\tilde{V}}\\
\label{apendeq1_2}&\Longleftrightarrow& \boldsymbol{\Sigma}_{W}^{-1} + \mathbf{K} \succeq \left(\mu-1\right) \boldsymbol{\Sigma}_{\tilde{V}}^{-1}  .
\end{eqnarray}
Since there always exists a non-singular matrix which simultaneously diagonalizes two positive semi-definite matrices \cite{Matrix:Horn}, there exists a non-singular matrix $\mathbf{Q}$ which simultaneously diagonalizes both $\boldsymbol{\Sigma}_{W}$ and $\boldsymbol{\Sigma}_{\tilde{V}}$ as follows:
\begin{eqnarray}
\label{apendeq2_1}\mathbf{Q}^T \boldsymbol{\Sigma}_{W} \mathbf{Q} & = & \mathbf{D}_{W},\\
\label{apendeq2_2}\mathbf{Q}^T \boldsymbol{\Sigma}_{\tilde{V}} \mathbf{Q} & = & \mathbf{I},
\end{eqnarray}
where $\mathbf{I}$ is an identity matrix, and $\mathbf{D}_W$ is a diagonal matrix. Since $\mathbf{Q}$ is a non-singular matrix, the inverse of $\mathbf{Q}$ always exists, and $\boldsymbol{\Sigma}_{W}$ and $\boldsymbol{\Sigma}_{\tilde{V}}$ are expressed as
\begin{eqnarray}
\label{apendeq3_1}\boldsymbol{\Sigma}_{W} & = & \mathbf{Q}^{-T}\mathbf{D}_W \mathbf{Q}^{-1},\\
\label{apendeq3_2}\boldsymbol{\Sigma}_{\tilde{V}} & = & \mathbf{Q}^{-T} \mathbf{Q}^{-1}.
\end{eqnarray}

If we define $\mathbf{D}_K$ as a diagonal matrix whose $i^{th}$ diagonal element is represented as $d_{K_i}$, and which it is defined as
\begin{equation}
\label{apendeq3_3}d_{K_i} =  \left\{ \begin{array} {ll} 0 & \text{if }d_{W_i} \leq \left(\mu-1\right)^{-1} \\
\mu-1- \frac{1}{d_{W_i}} & \text{if }d_{W_i} > \left(\mu-1\right)^{-1} \end{array} \right.
\end{equation}
where $d_{W_i}$ denotes the $i^{th}$ diagonal element of $\mathbf{D}_{W}$, and define $\mathbf{K}$ as
\begin{eqnarray}
\label{apendeq4_1}
\mathbf{K} & = & \mathbf{Q} \mathbf{D}_K \mathbf{Q}^{T},
\end{eqnarray}
then the equation (\ref{apendeq1_2}) is equivalent to
\begin{eqnarray}
\label{apendeq5_1}
&&\boldsymbol{\Sigma}_{W}^{-1} + \mathbf{K} \succeq \left(\mu-1\right) \boldsymbol{\Sigma}_{\tilde{V}}^{-1}\\
\label{apendeq5_2}
&\Longleftrightarrow& \left(\mathbf{Q}^{-T} \mathbf{D}_W \mathbf{Q}^{-1}\right)^{-1} + \mathbf{Q}\mathbf{D}_K \mathbf{Q}^{T} \succeq \left(\mu-1\right) \left(\mathbf{Q}^{-T} \mathbf{Q}^{-1}\right)^{-1}\\
\label{apendeq5_3}
&\Longleftrightarrow& \mathbf{D}_{W}^{-1} + \mathbf{D}_K \succeq \left(\mu-1\right) \mathbf{I}.
\end{eqnarray}
The equation (\ref{apendeq5_3}) always holds since $\mathbf{D}_K$ is defined in (\ref{apendeq3_3}). Therefore, the inequality (\ref{apendeq1_1}) is also satisfied.

We know that $\boldsymbol{\Sigma}_{X^*}$ is $(\mu-1)^{-1}\boldsymbol{\Sigma}_{\tilde{V}} -  \boldsymbol{\Sigma}_{\tilde{W}}$. Therefore,
\begin{eqnarray}
\label{apendeq6_1}
\boldsymbol{\Sigma}_{X^*} \mathbf{K} & = & \left(\mu-1 \right)^{-1} \left(\boldsymbol{\Sigma}_{\tilde{V}} - \left(\mu-1\right) \boldsymbol{\Sigma}_{\tilde{W}}\right) \mathbf{K},
\end{eqnarray}
and the equation (\ref{apendeq6_1}) is re-written as
\begin{eqnarray}
\label{apendeq7_1}
\boldsymbol{\Sigma}_{X^*} \mathbf{K} & = & \left(\mu-1 \right)^{-1} \left(\boldsymbol{\Sigma}_{\tilde{V}} - \left(\mu-1\right) \boldsymbol{\Sigma}_{\tilde{W}}\right) \mathbf{K}\\
\label{apendeq7_2}
& = & \left(\mu-1 \right)^{-1} \left(\mathbf{Q}^{-T}\mathbf{Q}^{-1} - \left(\mu-1\right)\left(\left( \mathbf{Q}^{-T}\mathbf{D}_{W}\mathbf{Q}^{-1}\right)^{-1}+\mathbf{Q}\mathbf{D}_K\mathbf{Q}^T \right)^{-1}\right) \mathbf{Q}\mathbf{D}_K \mathbf{Q}^{T}\\
\label{apendeq7_3}
& = & \left(\mu-1 \right)^{-1} \mathbf{Q}^{-T}\left(\mathbf{I} - \left(\mu-1\right) \left(\mathbf{D}_{W}^{-1}+\mathbf{D}_K \right)^{-1}\right) \mathbf{D}_K \mathbf{Q}^{T}\\
\label{apendeq7_5}
& = & \mathbf{0}.
\end{eqnarray}
The equality (\ref{apendeq7_2}) is due to the equations (\ref{apendeq3_1}), (\ref{apendeq3_2}), and (\ref{apendeq4_1}), and the equality (\ref{apendeq7_5}) is due to (\ref{apendeq3_3}). Similarly,
\begin{eqnarray}
 \mathbf{K} \boldsymbol{\Sigma}_{X^*}& = & \left(\mu-1 \right)^{-1} \mathbf{K}\left(\boldsymbol{\Sigma}_{\tilde{V}} - \left(\mu-1\right) \boldsymbol{\Sigma}_{\tilde{W}}\right) \nonumber\\
& = &  \left(\mu-1 \right)^{-1} \mathbf{Q}\mathbf{D}_K \mathbf{Q}^{T} \left(\mathbf{Q}^{-T}\mathbf{Q}^{-1} - \left(\mu-1\right)\left(\left( \mathbf{Q}^{-T}\mathbf{D}_{W}\mathbf{Q}^{-1}\right)^{-1}+\mathbf{Q}\mathbf{D}_K\mathbf{Q}^T \right)^{-1}\right) \nonumber\\
& = & \left(\mu-1 \right)^{-1} \mathbf{Q} \mathbf{D}_K\left(\mathbf{I} - \left(\mu-1\right) \left(\mathbf{D}_{W}^{-1}+\mathbf{D}_K \right)^{-1}\right)  \mathbf{Q}^{-1}\nonumber\\
& = & \mathbf{0}.\nonumber
\end{eqnarray}

Therefore, by defining $\boldsymbol{\Sigma}_{\tilde{W}}=(\boldsymbol{\Sigma}_W^{-1}+\mathbf{K})^{-1}$, we can make $\boldsymbol{\Sigma}_{\tilde{W}}$ satisfy
\begin{eqnarray}
\label{apendeq9_1}\boldsymbol{\Sigma}_{\tilde{W}} \preceq \left(\mu-1\right)^{-1} \boldsymbol{\Sigma}_{\tilde{V}},\hspace{5mm}
\label{apendeq9_2}\boldsymbol{\Sigma}_{X^*} \mathbf{K} =  \mathbf{K} \boldsymbol{\Sigma}_{X^*} =\mathbf{0},
\end{eqnarray}
and the proof is completed.
\begin{rem}
In Lemmas \ref{lem7} and  \ref{lem8}, we specify the structure of positive semi-definite matrices $\mathbf{L}$ and $\mathbf{K}$, and this yields additional details about  the structure of  the covariance matrix of the optimal solution.
\end{rem}

\section{A More Simplified Proof of the EEI}\label{appD}

The problem  in (\ref{EI_eq52_1})  is more appropriately re-formulated as follows:
\begin{eqnarray}
\label{EI_eq21_1}
\hspace{-10mm}&&\max_{f_{ \hat{X}}, f_{ Y}} \quad  \int \int f_{ X}(\mathbf{x}) f_{ \hat{V}}(\mathbf{y}-\mathbf{x}) \left(\mu\log f_{ Y}(\mathbf{y}) - \log f_{ \hat{X}}(\mathbf{x}) - \mu\left(\mu-1\right) \log f_{ \hat{V}}(\mathbf{y}-\mathbf{x})\right) d\mathbf{x} d\mathbf{y}\\
\hspace{-10mm}&&\text{s.t.}\quad\quad \int\int f_{ \hat{X}}(\mathbf{x}) f_{\hat{V}}(\mathbf{y}-\mathbf{x})d\mathbf{x}d\mathbf{y} = 1, \nonumber\\
\hspace{-10mm}&&\hspace{12mm} \int\int \left(y_i y_j -x_i x_j - \left(y-x\right)_i \left(y-x\right)_j \right)f_{ \hat{X}}(\mathbf{x}) f_{ \hat{V}}(\mathbf{y}-\mathbf{x})d\mathbf{x} d\mathbf{y} = 0, \nonumber\\
\hspace{-10mm}&&\hspace{12mm}\sum_{i=1}^{n}\sum_{j=1}^n\left( \int\int x_i x_j \xi_i \xi_j f_{ \hat{X}}(\mathbf{x}) f_{ \hat{V}}(\mathbf{y}-\mathbf{x})d\mathbf{x}d\mathbf{y} \right)\leq  \sum_{i=1}^{n}\sum_{j=1}^n\sigma^2_{ij} \xi_i \xi_j,\nonumber\\
\hspace{-10mm}&&\hspace{12mm} \int\int y_iy_j f_{ \hat{X}}(\mathbf{x}) f_{ \hat{V}}(\mathbf{y}-\mathbf{x})d\mathbf{x}d\mathbf{y} = {\sigma}^2_{ Y^*_{ij}}, \nonumber\\
\hspace{-10mm}&&\hspace{12mm} -\int \int f_{ \hat{X}}(\mathbf{x}) f_{ \hat{V}}(\mathbf{y}-\mathbf{x}) \log f_{ \hat{X}}(\mathbf{x}) d\mathbf{x} d\mathbf{y} = p_{ \hat{X}} ,\nonumber\\
\label{EI_eq21_5}
\hspace{-10mm}&&\hspace{12mm} f_{ Y}(\mathbf{y}) = \int\int f_{ \hat{X}}(\mathbf{x}) f_{ \hat{V}}(\mathbf{y}-\mathbf{x})d\mathbf{x}d\mathbf{y},
\end{eqnarray}
where the  arbitrary deterministic non-zero vector $\boldsymbol{\xi}$ is defined as $[\xi_1, \ldots, \xi_n]^T$, ${\sigma}^2_{\scriptscriptstyle Y^*_{ij}}$ denotes  the $i^{\text{th}}$ row and  $j^{\text{th}}$ column element of $\boldsymbol{\Sigma}_{\scriptscriptstyle Y^*}$, $i=1,\ldots,n$, and $j=1,\ldots,n$.

Using Lagrange multipliers, the functional problem and its constraints in (\ref{EI_eq21_1}) are expressed as
\begin{eqnarray}
\label{EI_eq22_1}
\max_{f_{ \hat{X}}, f_{ Y}} \quad  \int \left(\int K(\mathbf{x},\mathbf{y},f_{ \hat{X}}, f_{ Y}) d\mathbf{x}\right) + \tilde{K}(\mathbf{y},f_{ Y})d\mathbf{y},\nonumber\\
\end{eqnarray}
where
\begin{eqnarray}
\label{EI_eq23_1}
K(\mathbf{x},\mathbf{y},f_{ \hat{X}}, f_{ Y}) & = & f_{ \hat{X}}(\mathbf{x})f_{ \hat{V}}(\mathbf{y}-\mathbf{x}) \Big(\mu\log f_{ Y}(\mathbf{y})-\log f_{ \hat{X}}(\mathbf{x}) - \mu\left(\mu-1\right) \log f_{ \hat{V}}(\mathbf{y}-\mathbf{x}) + \alpha_0\nonumber\\
&& + \sum_{i=1}^{n} \sum_{j=1}^{n} \Big(  \gamma_{ij} y_i y_j -  \gamma_{ij} x_i x_j - \gamma_{ij} \left(y-x\right)_i \left(y-x\right)_j +  \theta x_i x_j \xi_i \xi_j + \phi_{ij} y_i y_j \Big)\nonumber\\
&&\hspace{20mm}- \alpha_1 \log f_{ \hat{X}}(\mathbf{x}) - \lambda(\mathbf{y})\Big),\nonumber\\
\tilde{K}(\mathbf{y},f_{ Y}) & = & \lambda(\mathbf{y})f_{ Y}(\mathbf{y}),
\end{eqnarray}
where $\alpha_0$, $\alpha_1$, $\gamma_{ij}$, $\theta$, $\phi_{ij}$, and $\lambda(\mathbf{y})$ stand for the  Lagrange multipliers.

The  first-order variation condition is checked as follows:
\begin{eqnarray}
\label{EI_eq24_1}
K'_{f_{ \hat{X}}}\Big|_{f_{ \hat{X}}=f_{ \hat{X}^*}, f_{ Y}=f_{ Y^*}} & = &  0 \\
\label{EI_eq24_2}
\tilde{K}'_{f_{ Y}}\Big|_{f_{ Y}=f_{ \hat{X}^*}, f_{ Y}=f_{ Y^*}} & = & 0,
\end{eqnarray}
$K'_{f_{ \hat{X}}}$ and $\tilde{K}'_{f_{ Y}}$  are the first-order partial derivatives with respect to $f_{ \hat{X}}$ and $f_{ Y}$, respectively.\footnote{Throughout the paper, the arguments of functionals or functions are omitted unless the arguments are ambiguous or confusing.}

Since the equalities in (\ref{EI_eq24_1}) and (\ref{EI_eq24_2}) must be satisfied for any $\mathbf{x}$ and $\mathbf{y}$, one can easily obtain the following Gaussian density functions $f_{ \hat{X}^*}$ and $f_{ Y^*}$  as solutions:
\begin{eqnarray}
\label{EI_eq27_1}
f_{ Y^*}(y) & = & \left(2\pi\right)^{-\frac{n}{2}} \left|\boldsymbol{\Sigma}_{ Y^*}\right|^{-\frac{1}{2}}\exp\left\{-\frac{1}{2}\mathbf{y}^{ T} \boldsymbol{\Sigma}_{ Y^*}^{-1}\mathbf{y}\right\},\nonumber\\
f_{ \hat{X}^*}(x) & = & \left(2\pi \right)^{-\frac{n}{2}}\left|\boldsymbol{\Sigma}_{ \hat{X}^*}\right|^{-\frac{1}{2}}\exp\left\{-\frac{1}{2}\mathbf{x}^{ T} \boldsymbol{\Sigma}_{ \hat{X}^*}^{-1}\mathbf{x}\right\}.
\end{eqnarray}
Since all the Lagrange multipliers exist in this problem, the necessary optimal solutions $f_{ \hat{X}^*}$ and $f_{ Y^*}$ exist even though the original problem is non-convex in general.

To make the second variation positive,  the negative-definiteness of the following matrix is required:
\begin{eqnarray}
\label{EI_eq34_1}
\left[\begin{array}{cc}
K''_{f_{ \hat{X}^*} f_{ \hat{X}^*}} &  K''_{f_{ \hat{X}^*} f_{ Y^*}}\\
K''_{f_{ Y^*} f_{ \hat{X}^*}} &  K''_{f_{ Y^*} f_{ Y^*}}
\end{array}\right],
\end{eqnarray}
where $K''_{f_{ \hat{X}^*} f_{ \hat{X}^*}}$ and $K''_{f_{ Y^*} f_{ Y^*}}$ stand for the second-order partial derivatives with respect to $f_{\hat{X}^*}$ and $f_{Y^*}$, respectively, and $K''_{f_{ \hat{X}^*} f_{ Y^*}}$ denotes  the second-order partial derivative with respect to $f_{ \hat{X}^*}$ and $f_{ Y^*}$. Thus,   the following condition is required to hold:
\begin{eqnarray}
\label{EI_eq35_1}
&&\left[\begin{array}{cc} h_{ \hat{X}} & h_{ Y}\end{array} \right] \left[\begin{array}{cc}
K''_{f_{ \hat{X}^*} f_{ \hat{X}^*}} &  K''_{f_{ \hat{X}^*} f_{ Y^*}}\\
K''_{f_{ Y^*} f_{ \hat{X}^*}} &  K''_{f_{ Y^*} f_{ Y^*}}
\end{array}\right]
\left[\begin{array}{c} h_{ \hat{X}} \\ h_{ Y}\end{array} \right] \nonumber\\
& = & K''_{f_{ \hat{X}^*} f_{ \hat{X}^*}} h_{ \hat{X}}^2 + K''_{f_{ Y^*} f_{ Y^*}} h_{ Y}^2 +  (K''_{f_{ \hat{X}^*} f_{ Y^*}}+K''_{f_{ Y^*} f_{ \hat{X}^*}}) h_{ Y} h_{ \hat{X}}\nonumber\\
& \leq & 0,
\end{eqnarray}
where $h_{ \hat{X}}$ and $h_{ Y}$ are arbitrary admissible functions.

Since $K''_{f_{ \hat{X}^*} f_{ \hat{X}^*}}$, $K''_{f_{ \hat{X}^*} f_{ Y^*}}$, $K''_{f_{ Y^*} f_{ \hat{X}^*}}$, and $K''_{f_{ Y^*} f_{ Y^*}}$ are defined as
\begin{eqnarray}
\label{EI_eq36_1}
K''_{f_{ \hat{X}^*} f_{ \hat{X}^*}} & = & -\frac{(1-\alpha_1)f_{ \hat{V}}(\mathbf{y}-\mathbf{x})}{f_{ X^*}(\mathbf{x})},\nonumber\\
K''_{f_{ \hat{X}^*} f_{ Y^*}} & = & \frac{\mu f_{ \hat{V}}(\mathbf{y}-\mathbf{x})}{f_{ Y^*}(\mathbf{y})},\nonumber\\
K''_{f_{ Y^*} f_{ \hat{X}^*}} & = & \frac{\mu f_{ \hat{V}}(\mathbf{y}-\mathbf{x})}{f_{ Y^*}(\mathbf{y})},\nonumber\\
K''_{f_{ Y^*} f_{ Y^*}} & = & -\frac{\mu f_{ X^*}(\mathbf{x}) f_{ \hat{V}}(\mathbf{y}-\mathbf{x})}{f_{ Y^*}(\mathbf{y})^2},
\end{eqnarray}
the equation in (\ref{EI_eq35_1}) requires
\begin{eqnarray}
\label{EI_eq37_1}
&& -\frac{(1-\alpha_1) f_{ \hat{V}} (\mathbf{y}-\mathbf{x})}{f_{ \hat{X}^*}(\mathbf{x})}h_{ \hat{X}}(\mathbf{x})^2 + 2\frac{\mu f_{ \hat{V}}(\mathbf{y}-\mathbf{x})}{f_{ Y^*}(\mathbf{y})}h_{ \hat{X}}(\mathbf{x})h_{ Y}(\mathbf{y}) - \frac{\mu f_{ \hat{X}^*}(\mathbf{x})f_{ \hat{V}}(\mathbf{y}-\mathbf{x})}{f_{ Y^*}(\mathbf{y})^2}h_{ Y}(\mathbf{y})^2\nonumber\\
& \leq & -\frac{\mu f_{ \hat{V}}(\mathbf{y}-\mathbf{x})}{f_{ \hat{X}^*}(\mathbf{x})} \left( h_{ \hat{X}}(\mathbf{x}) - \frac{f_{ \hat{X}^*}(\mathbf{x})}{f_{ Y^*}(\mathbf{y})}h_{ Y}(\mathbf{y}) \right)^2\nonumber\\
& \leq & 0,
\end{eqnarray}
where $\alpha_1 \geq 1-\mu$.

Therefore, the optimal solutions $f_{\hat{X}^*}$ and $f_{ Y^*}$ maximize the functional problem in (\ref{EI_eq52_1}), and the proof is completed.

\section{Details of an application for broadcasting channel with a private message}\label{appC}

Using Lemma \ref{lem8}, we can define a covariance matrix $\boldsymbol{\Sigma}_{\tilde{Z}_{G_1}}$ which satisfies $\boldsymbol{\Sigma}_{\tilde{Z}_{G_1}} \preceq \boldsymbol{\Sigma}_{Z_{G_1}}$ and $\boldsymbol{\Sigma}_{\tilde{Z}_{G_1}} \preceq \boldsymbol{\Sigma}_{{Z}_{G_2}}$ as follows:
\begin{eqnarray}
\label{appC_eq1_1}
\boldsymbol{\Sigma}_{\tilde{Z}_{G_1}} & = & \left(\boldsymbol{\Sigma}_{Z_{G_1}}+\mathbf{K}\right)^{-1},\nonumber
\end{eqnarray}
where $\mathbf{K}$ is a positive semi-definite matrix,  defined similarly to  the one in Lemma \ref{lem8}.

Since
\begin{eqnarray}
\label{appC_eq2_1}
\boldsymbol{\Sigma}_{X|Y_2} & = & \boldsymbol{\Sigma}_{Z_{G_2}} -  \boldsymbol{\Sigma}_{Z_{G_2}} \mathbf{J}\left(X+Z_{G_2}\right) \boldsymbol{\Sigma}_{Z_{G_2}},
\end{eqnarray}
where $\mathbf{J}\left(X+Z_{G_2}\right)$ denotes the Fisher information matrix of the random vector $X+Z_{G_2}$ \cite{EPI:Rioul}, by changing the covariance matrix of $X$, we can always find $X$, whose posterior covariance matrix $\boldsymbol{\Sigma}_{X|Y_2}$ satisfies $\mathrm{Tr}\{\boldsymbol{\Sigma}_{X|Y_2}\} = \mathrm{Tr}\{\mathbf{R}\}$.

Then the random vector $X$ satisfies the following relationship:
\begin{eqnarray}
\label{appC_eq3_1}
\mathrm{Tr}\{\boldsymbol{\Sigma}_{X|\tilde{Y}_1}\} \leq \mathrm{Tr}\{\boldsymbol{\Sigma}_{X|Y_2}\} = \mathrm{Tr}\{\mathbf{R}\},
\end{eqnarray}
where $\tilde{Y}_1 = X + \tilde{Z}_{G_1}$. The first inequality is due to the data processing inequality \cite{EPI:Rioul}.

Using Cram\'{e}r-Rao inequality \cite{EPI:Rioul}, we can choose a Gaussian random vector $X_G^*$, whose covariance matrix $\boldsymbol{\Sigma}_{X_G^*}$ satisfies the following:
\begin{eqnarray}
\label{appc_eq0_1}
\boldsymbol{\Sigma}_{X_G^*} + \boldsymbol{\Sigma}_{Z_{G_2}} & = & \mathbf{J}\left(X+Z_{G_2}\right)^{-1}\\
& \preceq & \mathbf{J}\left(X_G+Z_{G_2}\right)^{-1} \nonumber\\
& = & \boldsymbol{\Sigma}_{X_G} + \boldsymbol{\Sigma}_{Z_{G_2}}, \nonumber
\end{eqnarray}
and
\begin{eqnarray}
\label{appc_eq1_1}
\boldsymbol{\Sigma}_{X_G^*}  \preceq \boldsymbol{\Sigma}_{X_G}.
\end{eqnarray}
Therefore, for any random vector $X$, whose covariance matrix $\boldsymbol{\Sigma}_{X}$ satisfies $\mathrm{Tr}\{\boldsymbol{\Sigma}_{X|Y_2}\} = \mathrm{Tr}\{\mathbf{R}\}$, we can find a Gaussian random vector $X_G^*$, whose covariance matrix satisfies the relationship in (\ref{appc_eq1_1}). Also, due to the equations in (\ref{appC_eq2_1}) and (\ref{appc_eq0_1}),
\begin{eqnarray}
\boldsymbol{\Sigma}_{X_G^*|Y_2^*} = \boldsymbol{\Sigma}_{X|Y_2} ,\nonumber
\end{eqnarray}
where $Y_2^*=X_G^*+Z_{G_2}$.

Now, based on Lemma \ref{lem8}, we will show $\boldsymbol{\Sigma}_{X_G^*|\tilde{Y}_1}=\boldsymbol{\Sigma}_{X_G^*|{Y}_1}$ as follows.
Since $Y^*_1 = X_G^* + Z_{G_1} =  X_G^* + \tilde{Z}_{G_1} + \hat{Z}_{G_1}$, we can construct a Markov chain as
\begin{eqnarray}
\label{appC_eq6_1}
X_G^* \longrightarrow X_G^* + \tilde{Z}_{G_1}  \longrightarrow  X_G^* + \tilde{Z}_{G_1} + \hat{Z}_{G_1},
\end{eqnarray}
where $\hat{Z}_{G_1}$ is a Gaussian random vector with the covariance matrix $\boldsymbol{\Sigma}_{\hat{Z}_{G_1}}$, and it satisfies $\boldsymbol{\Sigma}_{{Z}_{G_1}}=\boldsymbol{\Sigma}_{\tilde{Z}_{G_1}}+\boldsymbol{\Sigma}_{\hat{Z}_{G_1}}$.

The Markov chain in (\ref{appC_eq6_1}) is the same as the one in (\ref{thm4e4_1}), and therefore, based on Lemma \ref{lem8}, we can obtain the  Markov chain:
\begin{eqnarray}
\label{appC_eq7_1}
X_G^* \longrightarrow  X_G^* + \tilde{Z}_{G_1} + \hat{Z}_{G_1}   \longrightarrow X_G^* + \tilde{Z}_{G_1},\nonumber
\end{eqnarray}
and this Markov chain is the same as the one in  (\ref{thm4e4_2}). In this case, $\boldsymbol{\Sigma}_{\tilde{Z}_{G_1}} = \left(\boldsymbol{\Sigma}_{Z_{G_1}}^{-1}+\mathbf{K}\right)^{-1}$, and $\boldsymbol{\Sigma}_{X_G^*}$ is defined as $\alpha \boldsymbol{\Sigma}_{\tilde{Z}_{G_2}} - \boldsymbol{\Sigma}_{\tilde{Z}_{G_1}}$, where $\tilde{Z}_{G_2}$ and $\hat{Z}_{G_2}$ are Gaussian random vectors with covariance matrices $\boldsymbol{\Sigma}_{\tilde{Z}_{G_2}}$ and $\boldsymbol{\Sigma}_{\hat{Z}_{G_2}}$, respectively, $Z_{G_2}=\tilde{Z}_{G_1}+\tilde{Z}_{G_2}+\hat{Z}_{G_2}$, $\boldsymbol{\Sigma}_{{Z}_{G_2}}=\boldsymbol{\Sigma}_{\tilde{Z}_{G_1}}+\boldsymbol{\Sigma}_{\tilde{Z}_{G_2}}+\boldsymbol{\Sigma}_{\hat{Z}_{G_2}}$, and all random vectors are independent of one another. The positive semi-definite matrix $\mathbf{K}$ is defined as the one in Lemma \ref{lem8}. The constant $\alpha$ must be chosen to satisfy  the equation in (\ref{appc_eq0_1}). By defining the matrix $\boldsymbol{\Sigma}_{\tilde{Z}_{G_2}}$ as follows:
\begin{eqnarray}
\boldsymbol{\Sigma}_{\tilde{Z}_{G_2}} = \left(\left(\boldsymbol{\Sigma}_{\tilde{X}} + \boldsymbol{\Sigma}_{{Z}_{G_2}}  \right)^{-1} + \mathbf{L}\right)^{-1} -  \boldsymbol{\Sigma}_{{X}} -  \boldsymbol{\Sigma}_{\tilde{Z}_{G_1}},\nonumber
\end{eqnarray}
where  matrix $\mathbf{L}$ is similarly defined as the one in Lemma \ref{lem7}, the existence of such $X_G^*$ is guaranteed.

Therefore, by choosing a Gaussian random vector $X_G^*$ as mentioned previously,
\begin{eqnarray}
\mathrm{Tr}\left\{\boldsymbol{\Sigma}_{X_G^*|Y_1^*}\right\} \preceq \mathrm{Tr}\left\{\boldsymbol{\Sigma}_{X_G^*|Y_2^*}\right\} = \mathrm{Tr}\left\{\mathbf{R}\right\},\nonumber
\end{eqnarray}
and the covariance matrix $\boldsymbol{\Sigma}_{X_G^*}$ is the minimum value with respect to the positive semi-definite partial ordering, and the proof is completed.

\end{document}